\documentclass[%
 amsmath,amssymb,
 aps,
fleqn,
]{revtex4-2}
\usepackage[english]{babel}
\usepackage{graphicx}% Include figure files
\usepackage{dcolumn}% Align table columns on decimal point
\usepackage{bm}% bold math
\usepackage{hyperref}% add hypertext capabilities
\usepackage{url}
\hypersetup{
	colorlinks=false, % no color on links
	bookmarksnumbered=true,
	pdfborder={0 0 0},
}
\usepackage{silence}
\usepackage{amsmath,amssymb,amsthm,braket,listings,physics,here,xcolor,comment,blkarray,ascmac}
\usepackage{mathtools}
\usepackage[normalem]{ulem}
\usepackage{subfigure}

\newtheoremstyle{break}
  {\topsep}{\topsep}%
  {\normalfont}{}%
  {\itshape}{.}%
  { }{\thmname{#1}\thmnumber{ #2}}%
\theoremstyle{break}

\newtheorem{proposition}{Proposition}

\newtheorem{lemma}{Lemma}

\newcommand{\pt}{\partial}

\def\A{\mathrm {A}}
\def\B{\mathrm {B}}
\def\C{\mathrm {C}}
\def\X{\mathrm {X}}
\def\eps{\varepsilon}

\begin{document}

\preprint{preprint number}

\title{Partial synchronization and community switching in phase-oscillator networks and its analysis based on a bidirectional, weighted chain of three oscillators}% Force line breaks with \\

\author{Masaki Kato}
\affiliation{%
The Department of Mathematical Informatics,
%Graduate School of Information Science and Technology,
The University of Tokyo, Tokyo, Japan.
%  This line break forced with \textbackslash\textbackslash
}%
%  \altaffiliation[Also at ]{Physics Department, XYZ University.}%Lines break automatically or can be forced with \\
\author{Hiroshi Kori}%
\affiliation{%
The Department of Mathematical Informatics,
%Graduate School of Information Science and Technology,
The University of Tokyo, Tokyo, Japan.
%  This line break forced with \textbackslash\textbackslash
}%
\affiliation{Department of Complexity Sciences and Engineering,
The University of Tokyo, Kashiwa, Chiba, Japan.}
%  \email{Second.Author@institution.edu}

\date{\today}% It is always \today, today,

\begin{abstract}
Complex networks often possess communities defined based on network connectivity. 
When dynamics undergo in a network, one can also consider dynamical communities; i.e., a group of nodes displaying a similar dynamical process.
We have investigated both analytically and numerically the development of dynamical community structure, where the community is referred to as a group of nodes synchronized in frequency, in networks of phase oscillators.
We first demonstrate that using a few example networks, the community structure changes when network connectivity or interaction strength is varied.
In particular, we found that community switching, i.e., a portion of oscillators change the group to which they synchronize, occurs for a range of parameters.
We then propose a three-oscillator model: a bidirectional, weighted chain of three Kuramoto phase oscillators, as a theoretical framework for understanding the community formation and its variation.
Our analysis demonstrates that the model shows a variety of partially synchronized patterns: oscillators with similar natural frequencies tend to synchronize for weak coupling, while tightly connected oscillators tend to synchronize for strong coupling.
We obtain approximate expressions for the critical coupling strengths by employing a perturbative approach in a weak coupling regime and a geometric approach in a strong coupling regimes.
Moreover, we elucidate the bifurcation types of transitions between different patterns. 
Our theory might be useful for understanding the development of partially synchronized patterns in a wider class of complex networks than community structured networks. 
  
% \begin{description}
% Secondary publications and information retrieval purposes.
% You may use the \texttt{description} environment to structure your abstract;
% use the optional argument of the \verb+\item+ command to give the category of each item. 
% \end{description}
\end{abstract}


\maketitle

%\tableofcontents

\section{Introduction}
\label{sec:Introduction}

For decades now, complex networks and their dynamical processes are attracting much attention because of their broad applicability, which is ranges from biology to engineering \cite{RevModPhys.74.47,BOCCALETTI2006175}.
Complex networks often possess community (or module) structure, that is, there are groups of nodes that interact more intensely inside each group than between different groups \cite{girvan2002}.
In biological, social, and technological networks, the presence of a community takes an important role in their function.
For example, in the worldwide web, community may correspond to a set of webpages on related topics \cite{Flake2002};
in protein-protein interaction networks, it may correspond to a functional group involved in a particular process, such as cancer \cite{Taylor2009}.
Another notion of community is possible when dynamics are considered, that is, a group of nodes that share similar dynamical processes.
Dynamical community can be found in many systems, including in decision-making\cite{Couzin2005}, voter\cite{liggett1999stochastic,PhysRevE.80.041129}, and animal group consensus models\cite{Nabet2009,science.1210280,pnas.1118318108}.

A large body of dynamical community examples can be found in oscillator networks \cite{Louzada2012,PhysRevLett.106.054102,PhysRevE.84.046202,Nabet2009,pnas.1118318108,PhysRevE.70.056125,10.1143/ptp/86.6.1159,PhysRevLett.90.014101} 
For example, a group of oscillators with the same effective frequency (i.e., frequency-locked oscillators) can be regarded as a dynamical community\cite{PhysRevE.91.022901,Motter2013}.  %KUPTSOV2017115,
Such a dynamical community may emerge as a self-organization process, as experimentally demonstrated using interacting chemical oscillators \cite{Mikhailovpnas2004}.
In general, the formation of a community is expected to be facilitated in a group of oscillators with both intense connection and similar natural frequencies.
However, a group of oscillators does not necessarily share both properties. In this case, it is nontrivial which property dominates in the formation of a dynamical community.
Recently, it was numerically shown that in oscillators in a unidirectional chain, the community configuration depends on the overall coupling strength.
In particular, oscillators may go back and forth the community to which it belongs as the overall coupling strength increases \cite{XiaHuang:130506}. 
Here, we refer to this phenomenon as community switching. 
Although this work is of great interest because it demonstrates the plasticity of a community structure in oscillator networks, the mechanism and the conditions of the formation and the switching of a community remain obscure because their theoretical analysis is based on a very simple setup (i.e., a chain of three oscillators coupled uni-directionally).

In the present paper, we first numerically demonstrate that community switching occurs in oscillator networks when the network connectivity or the overall coupling strength is varied.
We then propose a theoretical framework for understanding community switching. 
More concretely, we propose and analyze a three-oscillator model comprising three oscillators connected in a bi-directional chain with arbitrary weights.
We obtain the expressions for the critical coupling strength in terms of the network parameters, natural frequencies, and overall coupling by partly exploiting a perturbative and geometric approaches. 
We also elucidate the bifurcations that occur when the community structure changes. 
Our analysis sheds light on the interplay between network connectivity and intrinsic dynamical properties in the community formation. 

The remainder of this paper is organized as follows: 
In Sec.~\ref{sec:method-3body-settting}, we conduct numerical simulations of particular oscillator networks and show that various community structures arise when the network connectivity or the overall coupling strength is manipulated.
In Sec.~\ref{sec:model-analysis}, we propose a three-oscillator model and analyze its basic properties. 
In Sec.~\ref{sec:3body-critical}, we present the bifurcation analysis of the three-oscillator model.
In particular, we conduct a higher-order perturbation analysis where the overall coupling strength and one of the frequency differences are treated as small parameters.
Moreover, we implement a geometric approach to prove the existence of periodic orbits corresponding to partial synchronizations.
In Sec.\ref{sec:applicationEg}, we apply our three-oscillator model to a particular oscillator network and show that our three-oscillator model qualitatively captures the behavior of a complex network.
Finally, in Sec.\ref{sec:concludediscuss}, we conclude and discuss our results.

\section{Synchronization and community switching in oscillator networks}
\label{sec:method-3body-settting}
We consider a network composed of $N$ Kuramoto phase oscillators \cite{Kuramoto1984}.
The model is given by a set of evolution equations for the oscillator phase $\phi_i$ ($i=1,2,\ldots,N$): 
\begin{equation}
    \dot{\phi}_i=\omega_i+\lambda\sum_{j=1}^N A_{ij}\sin(\phi_j-\phi_i),
\label{full_model}
\end{equation}
where $\omega_i$ is the natural frequency of the oscillator $i$, and $\lambda$ is the coupling strength.
The network connection topology is determined by an adjacency matrix $A$, 
where its element $A_{ij}$ is $1$ if oscillator $i$ is connected to oscillator $j$, and $0$ otherwise. 
In this article, we consider undirected networks only; thus, the matrix $A$ is symmetric.
We observe the effective frequency of the oscillator $i$, defined as follows:
\begin{equation}
 \Omega_i^{\mathrm{eff}}\coloneqq \frac{\phi_i(T)-\phi_i(T_0)}{T-T_0},
  \label{eff}
\end{equation}
where $T$ and $T_0$ are sufficiently large constants.
When the effective frequencies of oscillators $i$ and $j$ converge to a same value, i.e., $\Omega_i^\mathrm{eff}-\Omega_j^\mathrm{eff} \to 0$ as $T\to\infty$, these oscillators are regarded as being synchronized and thus in the same dynamical community.
We employ $T=3000$ and $T_0=1500$ in numerical simulations.

In this section, we numerically demonstrate the community switching in particular networks.
In the demonstration, we assume that $\omega_i$ takes only two values.
With this choice, we have two trivial communities in the absence of an interaction;
thus we can highlight the effect of the interaction on the community formation.
We investigate two particular cases: (i) a randomly evolving network with fixed $\lambda$; and (ii) various $\lambda$ values with a fixed network structure.
Their detailed analysis is performed in Sec.~\ref{sec:applicationEg} based on the theory presented in Sec.~\ref{sec:model-analysis}. 

\emph{Case (i)}.
We consider a complete network of $6$ oscillators, where the total number of edges is $l=15$.
We cut one randomly chosen edge in each step until all edges are removed.
Figure~\ref{fig:cutting_N6K10} shows the numerical simulation results.
Here, oscillator $1$ undergoes community switching between $l=6$ and $l=7$: whereas oscillator $1$ synchronizes with oscillators $0$ at $l=6$, it synchronizes with oscillator $2,3,4$ and $5$ at $l=7$.

\emph{Case (ii)}.
Figure~\ref{fig:changingK_N6} shows the numerical simulation results for a network composed of $6$ oscillators with the adjacency matrix given in Eq.~\eqref{eq:adjmat}.
In this example network, oscillator $1$ undergoes community switching at around $\lambda=0.6$.

Note that we do not observe hysteresis in both cases, suggesting only one stable state at each parameter set.

\begin{figure}[tbp]
  \begin{minipage}[t]{0.45\linewidth}
    \centering
    \includegraphics[height=50mm]{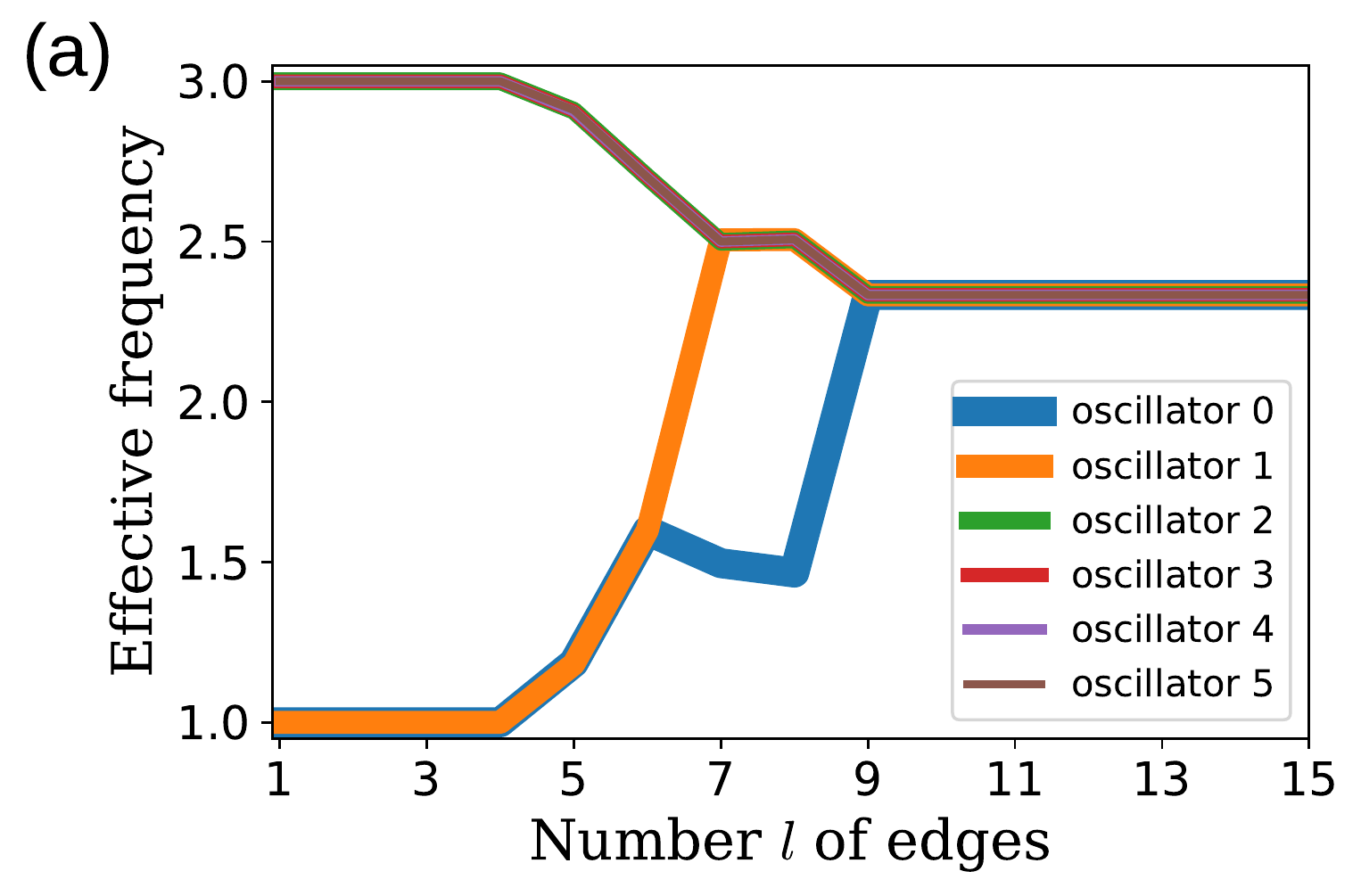}
    \refstepcounter{figure}
    \refstepcounter{subfigure}
    \label{fig:cutting_N6K10}
  \end{minipage}
  \begin{minipage}[t]{0.45\linewidth}
    \centering
    \includegraphics[height=50mm]{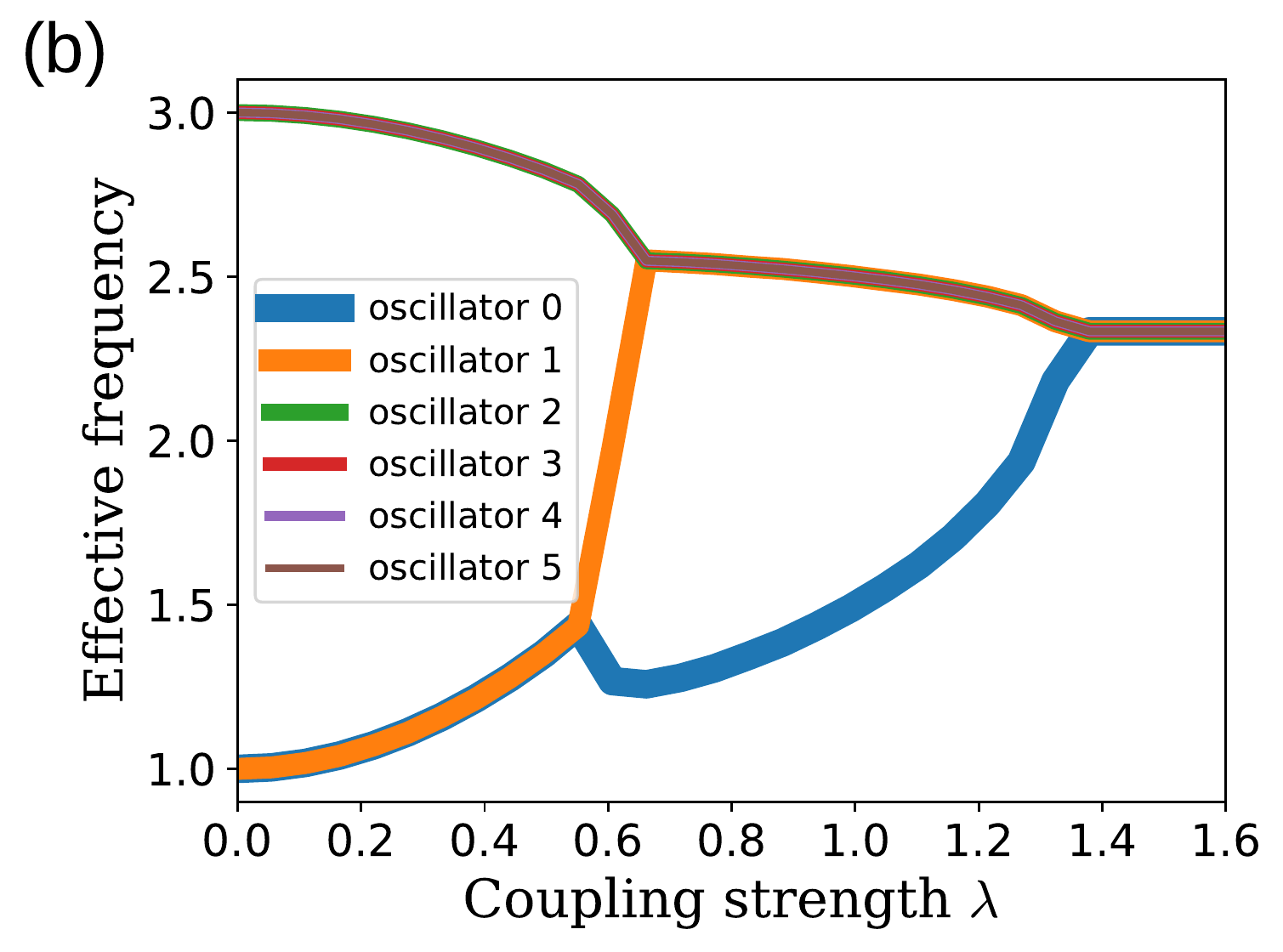}
    \refstepcounter{subfigure}
    \label{fig:changingK_N6}
  \end{minipage}
  \addtocounter{figure}{-1}
  \caption{   
  Community switching demonstrated in the full model, given in Eq.~\eqref{full_model}
  with $N=6$. We set $\omega_{0,1}=1$ and $\omega_{2,3,4,5}=3$.
  (a) A randomly evolving network comprised of $6$ oscillators with fixed $\lambda=1$.
  Community switching occurs between $l=6$ and $l=7$.
  (b) Various coupling strength $\lambda$ for a fixed adjacency matrix given in Eq.~\ref{eq:adjmat}.
  Community switching occurs at around $\lambda=0.6$.
  }
  \label{fig:example}
\end{figure}

\section{Three-oscillator model}
\label{sec:model-analysis}

\subsection{Model and its derivation}
\label{eq:model-dev}
As a tractable model useful for understanding and analyzing community switching,
we propose the following three-oscillator model:
\begin{subequations}
  \label{eq:3body}
  \begin{align}
    \dot{\phi}_\A&=\omega_\A+a r K \sin\left( \phi_\B-\phi_\A \right), \\
    \dot{\phi}_\B&=\omega_\B+a K\sin\left( \phi_\A-\phi_\B \right) + K\sin\left( \phi_\C-\phi_\B \right),\\
    \dot{\phi}_\C&=\omega_\C+ s K \sin\left( \phi_\B-\phi_\C \right),
  \end{align}
\end{subequations}
where $\phi_\X$ and $\omega_\X$ ($\X = \A, \B, \C$) are the phase and the intrinsic frequency of oscillator $\X$, respectively;
$K$ is the overall coupling strength; and $a, r,$ and $s$ are the parameters. 
Figure.~\ref{fig:model}(a) depicts the model schematics.
This model can reproduce the community switching of oscillator B induced by a variation in parameter values.

\begin{figure}[tbp]
   \centering
   \includegraphics[width=80mm]{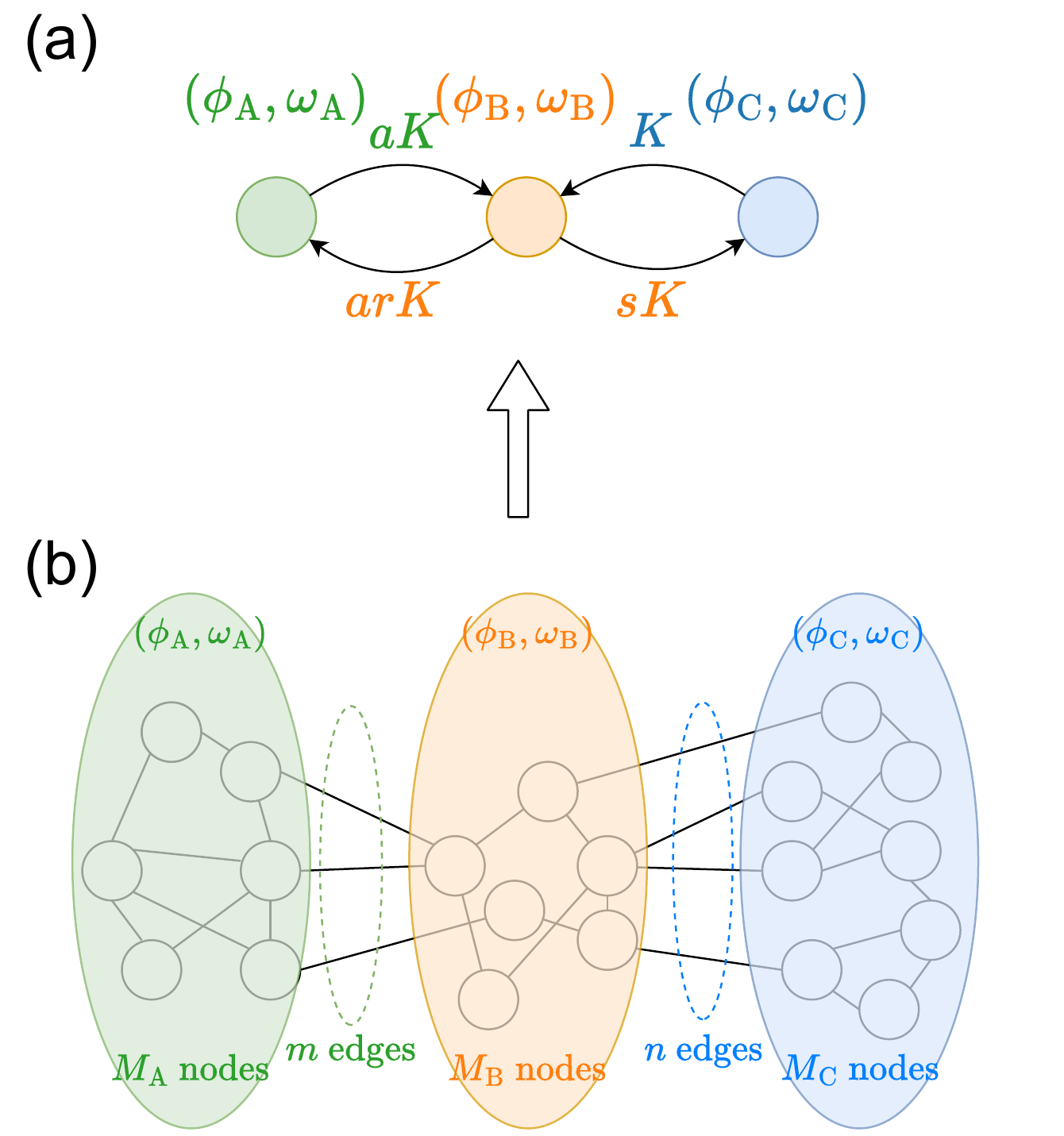}\\
   \caption{
    Model schematics. (a) Three-oscillator model given in Eq.~\eqref{eq:3body}. 
    (b) Three-group network: our three-oscillator model can be regarded as a mean-field approximation of the three-group network. 
    The correspondence between the parameters are given as $K= \frac{n\lambda}{M_\B}, a=\frac{m}{n}, r=\frac{M_\A}{M_\B}$, and $s=\frac{M_\B}{M_\C}$.
   }
  \label{fig:model}
\end{figure}

This model can be derived as a mean-field approximation of a particular type of oscillator network, the schematics of which is shown in Fig.~\ref{fig:model}(b).
We assume that the network comprises three groups $\A$, $\B$, and $\C$, and the oscillators are much more strongly connected inside each group than between groups.
Therefore, the phases of the oscillators in each group are assumed to be nearly the same.
Now we consider the following mean-field approximation: the properties of the oscillators inside each group are averaged and the entire dynamics are described as follows:
\begin{subequations}
  \label{eq:3body-default}
  \begin{align}
    \dot{\phi}_\A&=\omega_\A+\frac{m}{M_\A}\lambda\sin\left( \phi_\B-\phi_\A \right), \\
    \dot{\phi}_\B&=\omega_\B+\frac{n}{M_\B}\lambda\sin\left( \phi_\C-\phi_\B \right)+\frac{m}{M_\B}\lambda\sin\left( \phi_\A-\phi_\B \right),\\
    \dot{\phi}_\C&=\omega_\C+\frac{n}{M_\C}\lambda\sin\left( \phi_\B-\phi_\C \right),
  \end{align}
\end{subequations}
where $\phi_\X$, $\omega_\X$, $M_\X$ ($\X=\A,\B,\C$) is the averaged phase, averaged frequency, and number of oscillators in each group, respectively; $\lambda$ is the original coupling strength; and $m(n)$ is the number of edges between $\B$ and $\A(\C)$.
Introducing the overall coupling strength $K= \frac{n\lambda}{M_\B}$ and the ratios $a=\frac{m}{n}, r=\frac{M_\B}{M_\A}$, and $s=\frac{M_\B}{M_\C}$, we obtain Eq.~\eqref{eq:3body}.
A variation in $a$ in the three-oscillator model is associated with a change in the network connectivity in the original network.
The application of the three-oscillator model to the example networks presented in Fig.~\ref{fig:3body-application} is presented in Sec.~\ref{sec:applicationEg}

Introducing phase differences $x\coloneqq \phi_\A-\phi_\B$ and $y\coloneqq \phi_\C-\phi_\B$, our three-oscillator model is further reduced to 
\begin{subequations}
  \label{eq:phase-diff}
  \begin{align}
    \dot{x}&=\nu_x-K \left\{ a \left( 1+r \right) \sin x+  \sin y \right\}, \label{eq:phase-diff-x}\\
    \dot{y}&=\nu_y-K \left\{ a \sin x+ \left( 1 + s \right) \sin y \right\}, \label{eq:phase-diff-y}
  \end{align}
\end{subequations}
where $\nu_x=\omega_\A-\omega_\B$ and $\nu_y=\omega_\C-\omega_B$. As far as we are concerned with the community structure, it is sufficient to analyze this two-dimensional system.
Without loss of generality, we assume $|\nu_x| \ge |\nu_y|$ and $\nu_x>0$.
The former is trivial because one may exchange the name of oscillators A and C when
$|\nu_x| < |\nu_y|$.
The latter comes from the fact that Eq.~\eqref{eq:phase-diff} is invariant
under the transformation $(x,y,\nu_x,\nu_y) \to (-x,-y,-\nu_x,-\nu_y)$.

\subsection{Synchronization patterns and transitions}
\label{sec:3body-synchro-state}
We have at most five synchronization patterns for the three oscillators: 
$\{\A,\B,\C\}$, $\{\A\B,\C \}$, $\{ \A\C,\B\}$, $\{\A,\B\C\}$, and $\{\A\B\C\}$, where each element represents the group of synchronized oscillators.
However, in our three-oscillator model \eqref{eq:3body}, it is expected that pattern $\{ \A\C,\B\}$ does not generically arise because oscillators A and C are connected via oscillator B.

Figure.~\ref{fig:3body-abs} illustrates some typical transition behaviors among the synchronization patterns for a variation in $K$, while the other parameters are fixed. 
The example shown in Fig.~\ref{fig:3body-notapprox175} contains all the four possible patterns.
We observe transitions $\{\A,\B,\C\}$ $\to \{\A, \B\C \}$ $\to \{\A,\B,\C\}$ $\to \{ \A\C,\B\}$, $\to \{\A\B\C\}$ as $K$ increases.
Thus four critical $K$ values characterize each transition.
As shown in Fig.~\ref{fig:3body-notapprox150} and \ref{fig:3body-notapprox200}, some transitions disappear for other $a$ values.
These transition behaviors are robust against small changes in the parameter values.
Qualitatively the same transition behaviors are observed, even when either $\nu_x$ or $\nu_y$ is very small.
Figure.~\ref{fig:3body-typical} presents an example, where we observe the same transition behavior as that in Fig.~\ref{fig:3body-notapprox175}.
Non-generic transitions (i.e., the transition between $\{\A,\B,\C\}$ and $\{\A\B\C\}$ at a critical $K$) may also arise for a special parameter set, such as $\nu_x = -\nu_y$ and $r=s=1$.
Note that in Fig. \ref{fig:3body-abs}, we tested the two cases of increasing and decreasing $K$ without resetting the $x$ and $y$ values and found no hysteresis.
Note also that in Figs.~\ref{fig:3body-notapprox175}, \ref{fig:3body-notapprox200}, and \ref{fig:3body-typical}, the effective frequencies non-monotonically vary as $K$ increases. 
As detailed in Appendix.\ref{appendix:mn-sync}, this non-monotonic behavior stems from $m:n$ synchronization between $x$ and $y$. 

Henceforth, we focus only on the generic transitions and analyze them.
We denote the four critical coupling strengths observed in Fig.~\ref{fig:3body-notapprox175} or (d) as $K_1$, $K_2$, $K_3$ and $K_4$ in ascending order.

% \captionsetup[figure]{justification=centering}
\begin{figure}[tbp]
    \begin{minipage}[b]{0.47\linewidth}
        \centering
        \includegraphics[keepaspectratio, scale=0.455]{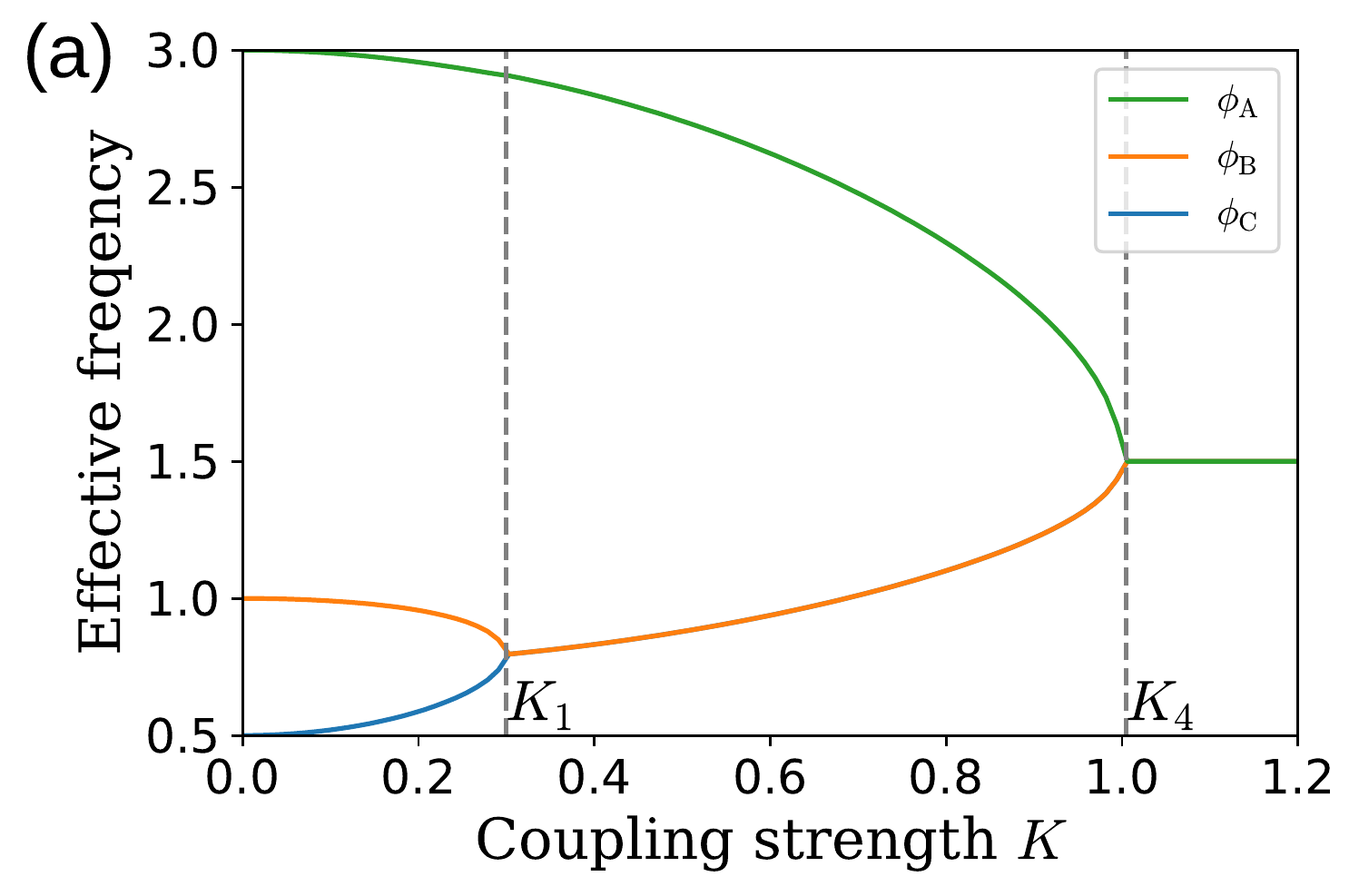}
        \refstepcounter{subfigure}
        \label{fig:3body-notapprox150}
    \end{minipage}
    \begin{minipage}[b]{0.47\linewidth}
        \centering
        \includegraphics[keepaspectratio, scale=0.455]{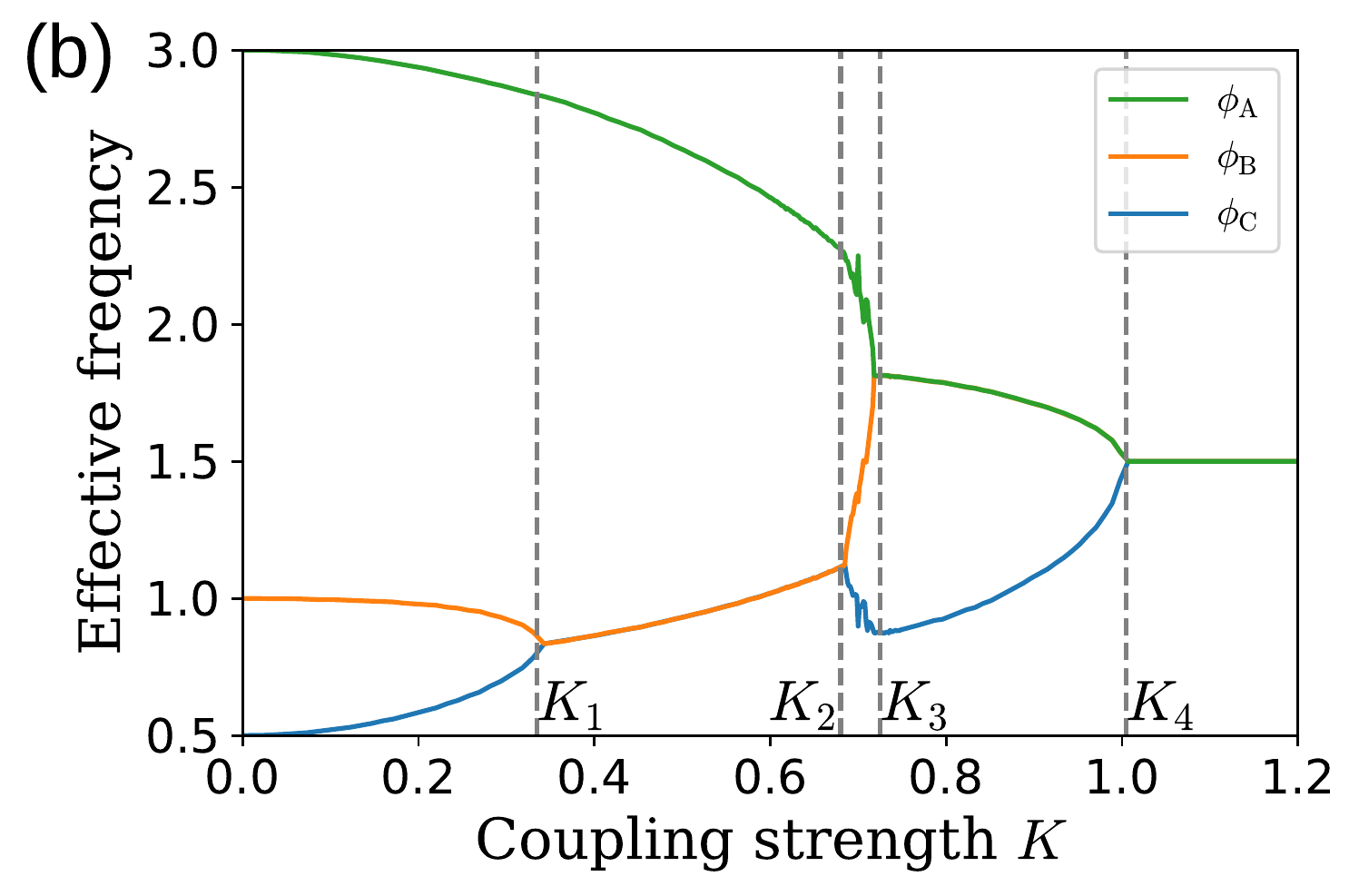}
        \refstepcounter{subfigure}
        \label{fig:3body-notapprox175}
    \end{minipage}\\
    \begin{minipage}[b]{0.47\linewidth}
        \centering
        \includegraphics[keepaspectratio, scale=0.455]{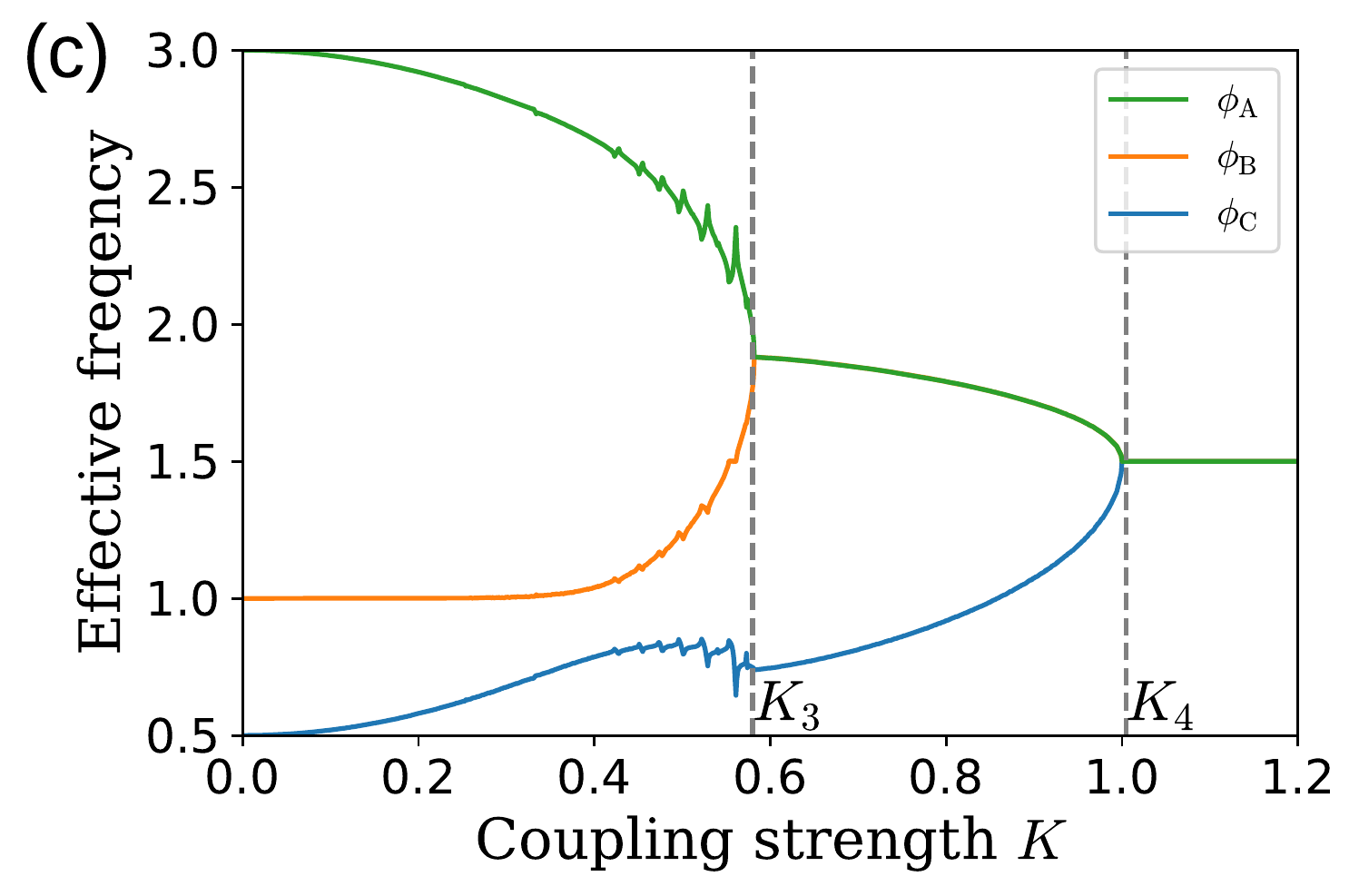}
        \refstepcounter{subfigure}
        \label{fig:3body-notapprox200}
    \end{minipage}
    \begin{minipage}[b]{0.47\linewidth}
        \centering
        \includegraphics[keepaspectratio, scale=0.455]{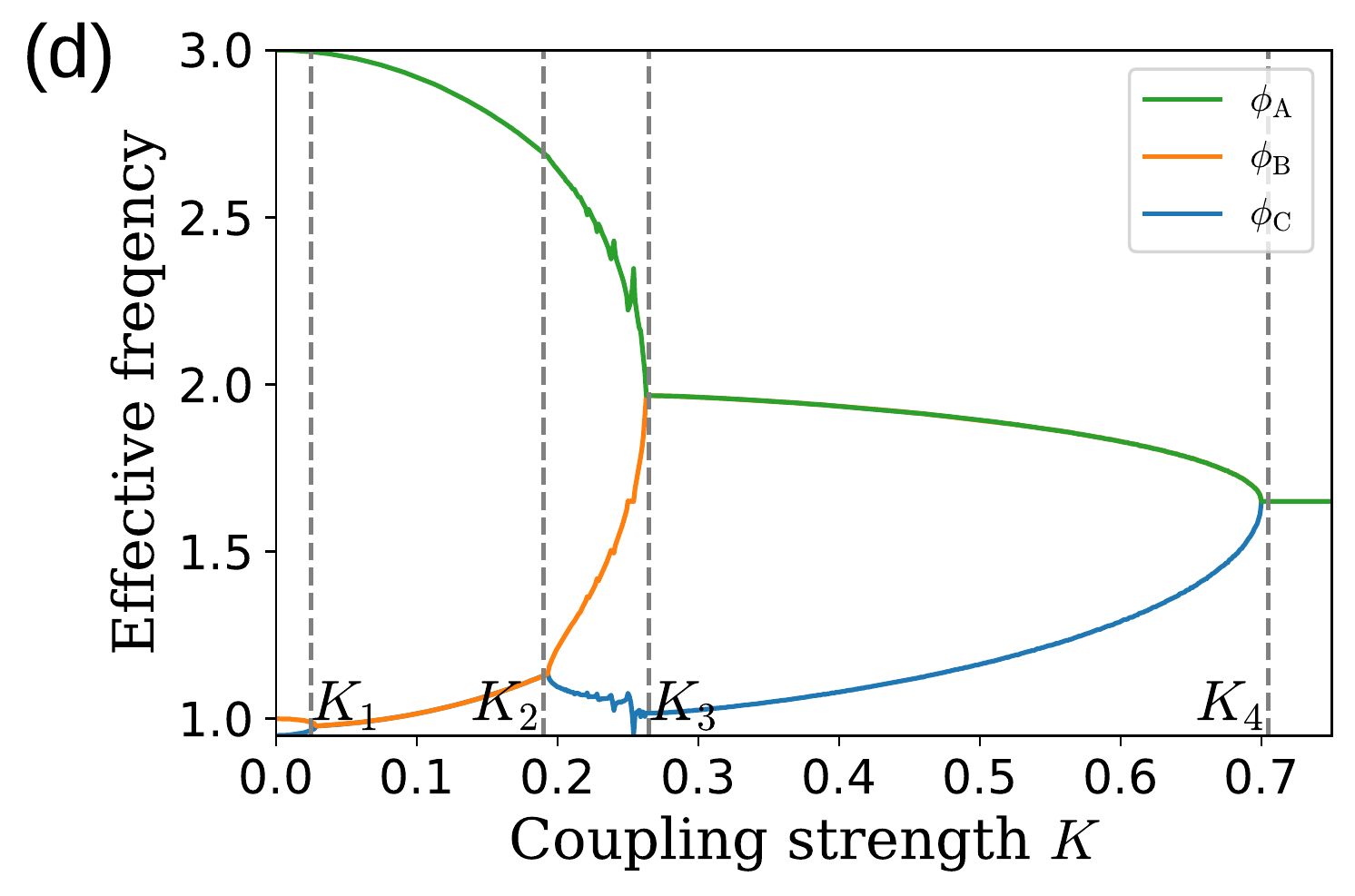}
        \refstepcounter{subfigure}
        \label{fig:3body-typical}
    \end{minipage}
    \caption{
    Synchronization patterns and transition behavior observed in the three-oscillator model given by Eq.~\eqref{eq:3body} with fixed $r=s=1$.
    (a-c) We fix $\omega_\A=3.0$, $\omega_\B=1.0$, and $\omega_\C=0.5$, whereas $a$ assumes different values: (a) $a=1.5$, (b) $a=1.75$, and (c) $a=2.0$. 
    (d) $\omega_\A=3.0$, $\omega_\B=1.0$, $\omega_\C=0.95$, and $a=4.0$.
    }
    \label{fig:3body-abs}
\end{figure}
% \clearcaptionsetup{figure}

\subsection{Phase plane analysis}
\label{sec:3body-phase-plane}
Each panel in Fig.~\ref{fig:phase} displays the vector field, $x$-nullcline, $y$-nullcline, sink, and trajectory out of $(x,y)=(0,0)$ of Eq.~\eqref{eq:phase-diff} for fixed $a=4.0$ and different $K$ values.
The figure caption describes the correspondence between the panel and the synchronization pattern.

Transitions (a)$\to$(b) and (c)$\to$(d) are characterized by the appearance of a stable limit cycle.
Similarly, (b)$\to$(c) is characterized by the disappearance of the stable limit cycles. 
These transitions are likely to be characterized by the saddle-node bifurcation of limit cycles.
In contrast, transition (d)$\to$(e) is characterized by the disappearance of the limit cycle and the appearance of a stable node, suggesting the saddle-node bifurcation on an invariant cycle.
These observations will be verified in Sec.~\ref{sec:3body-critical}.

\begin{figure}[tbp]
    \begin{minipage}[b]{0.47\linewidth}
        \centering
        \includegraphics[keepaspectratio, scale=0.45]{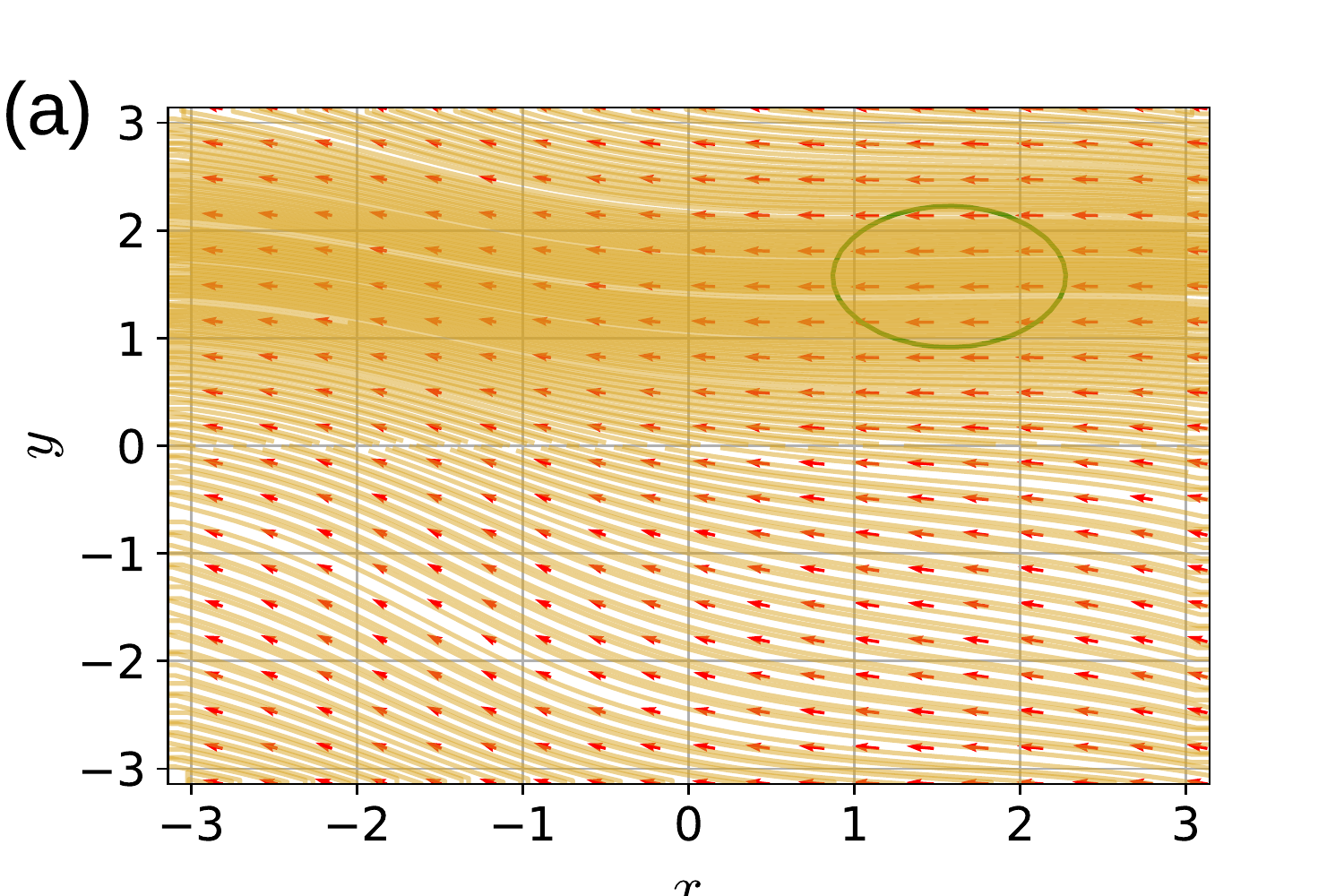}
        \refstepcounter{subfigure}
        \label{fig:phase-k2}
    \end{minipage}
    \begin{minipage}[b]{0.47\linewidth}
        \centering
        \includegraphics[keepaspectratio, scale=0.45]{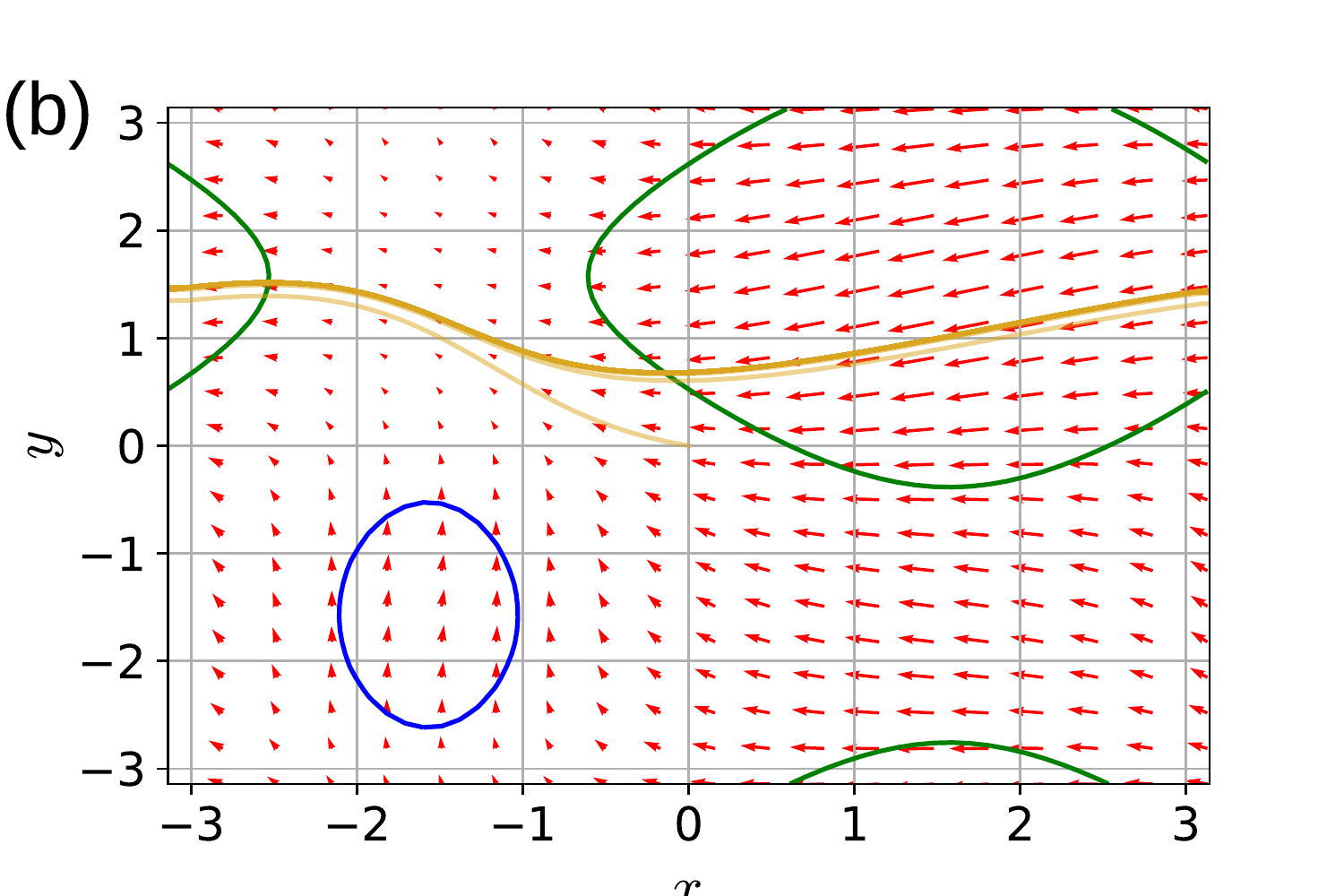}
        \refstepcounter{subfigure}
        \label{fig:phase-k10}
    \end{minipage}\\
    \begin{minipage}[b]{0.47\linewidth}
        \centering
        \includegraphics[keepaspectratio, scale=0.45]{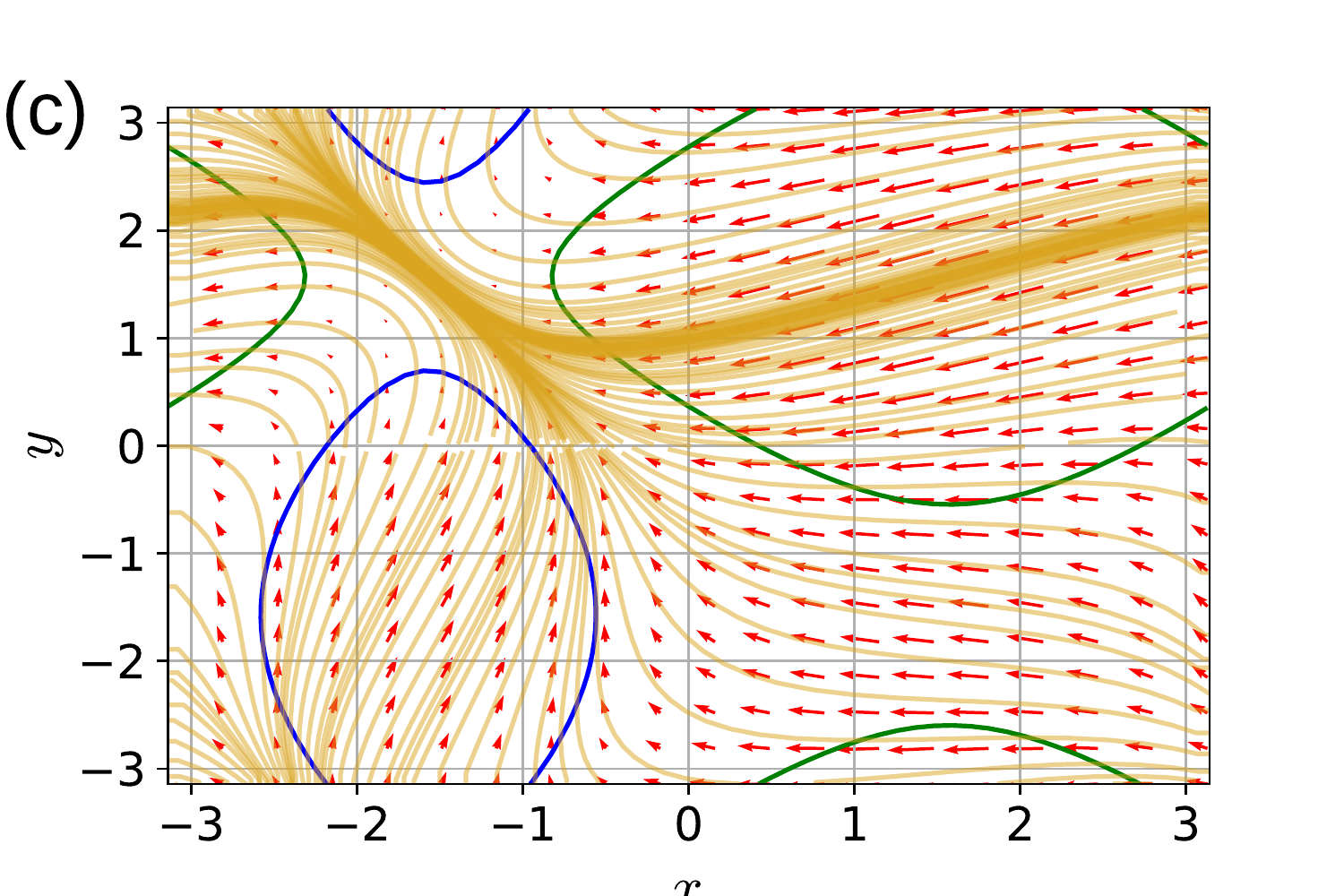}
        \refstepcounter{subfigure}
        \label{fig:phase-k22}
    \end{minipage}
    \begin{minipage}[b]{0.47\linewidth}
        \centering
        \includegraphics[keepaspectratio, scale=0.45]{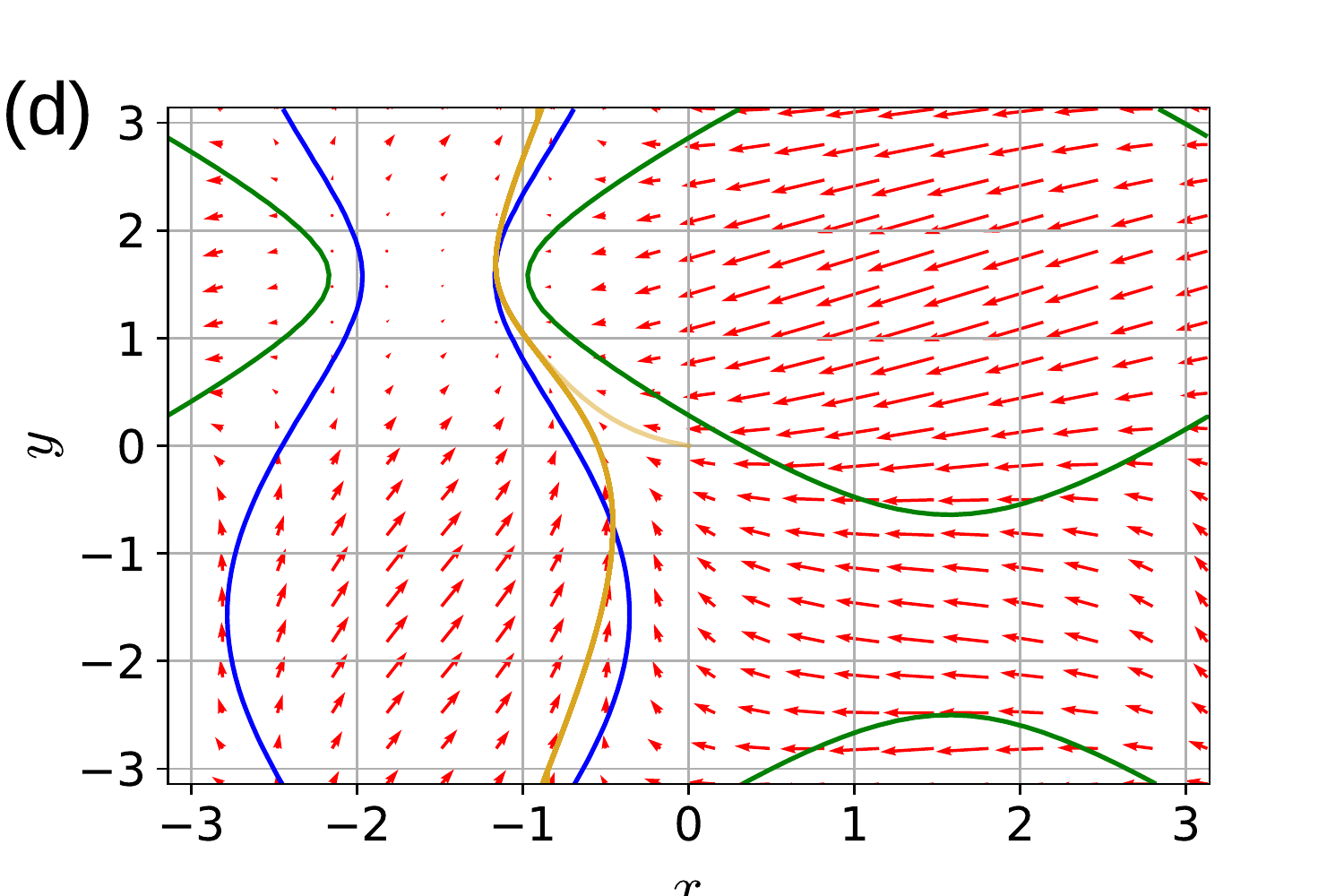}
        \refstepcounter{subfigure}
        \label{fig:phase-k50}
    \end{minipage}\\
    \begin{minipage}[b]{0.47\linewidth}
        \centering
        \includegraphics[keepaspectratio, scale=0.45]{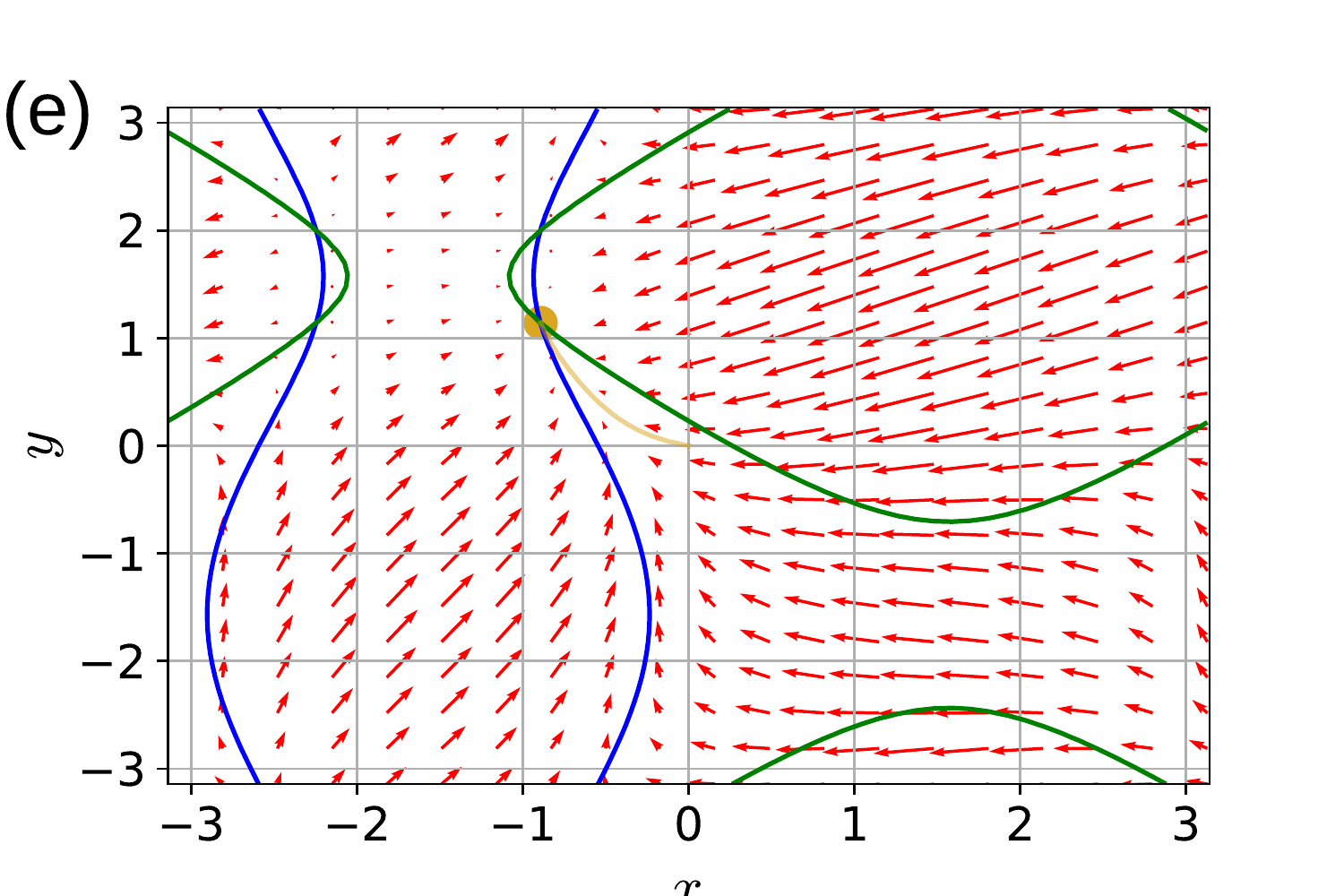}
        \refstepcounter{subfigure}
        \label{fig:phase-k80}
    \end{minipage}
    \begin{minipage}[b]{0.47\linewidth}
        \centering
        \includegraphics[keepaspectratio, scale=0.455]{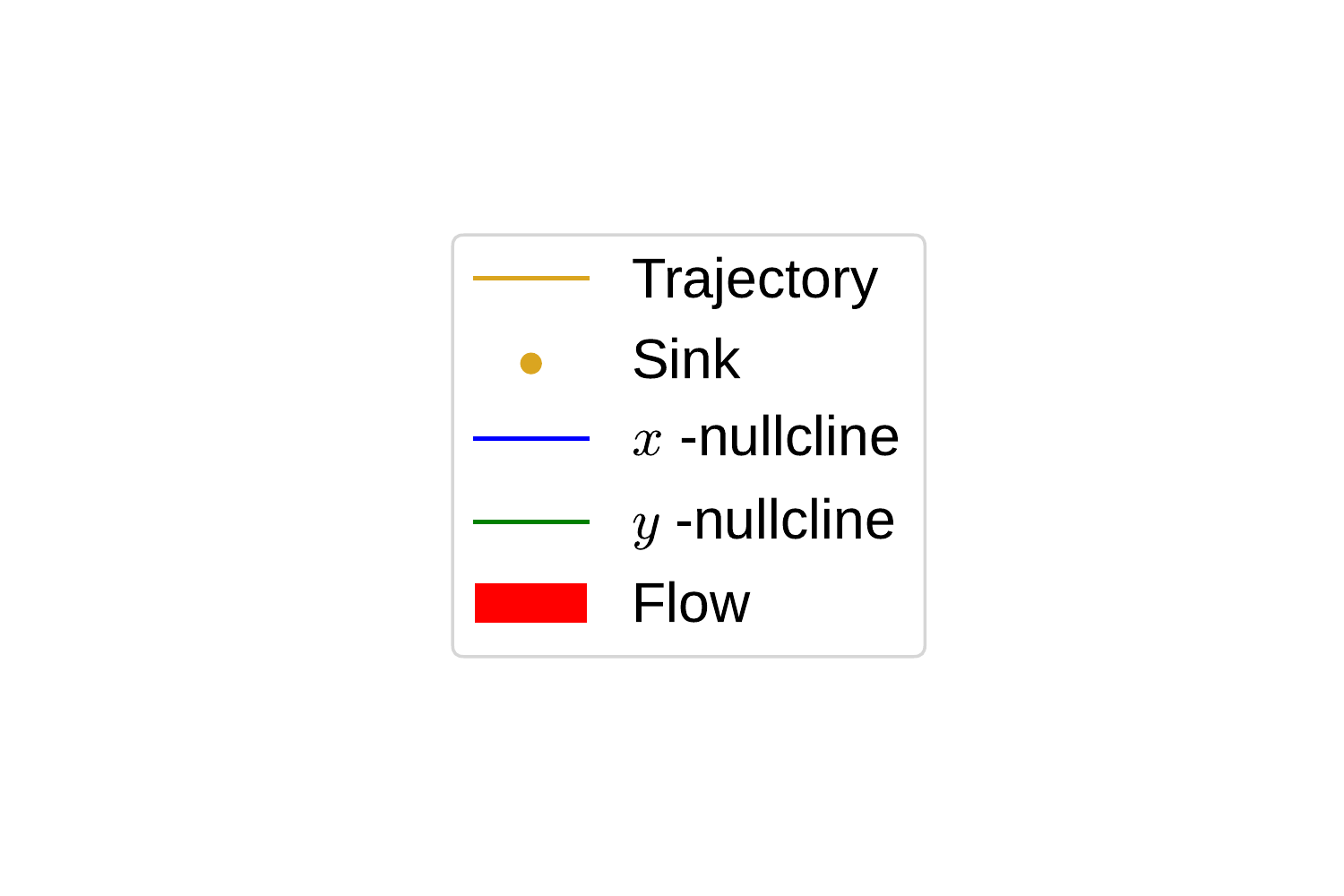}
        \refstepcounter{subfigure}
    \end{minipage}
    \caption{
      Typical trajectories in the phase plain in Eq.~\eqref{eq:phase-diff} for different $K$ and fixed $a=2.0$, $\nu_x=2$, and $\nu_y=-0.05$. 
      Each panel also includes the vector field, $x$-nullcline, $y$-nullcline, and sink.  
      (a) $K=0.15$. (b) $K=0.5$. (c) $K=0.7$. (d) $K=0.9$. (e) $K=1.1$.
      We observe \{$\A, \B, \C$\} in (a) and (c). Each corresponds to quasi-periodic solutions homeomorphic to $\mathbb{T}^2$.
      In (b) and (d), we observe \{$\A, \B\C$\} and  \{$\A\B, \C$\}, respectively, and each corresponds to stable limit cycles homeomorphic to $\mathbb{S}^1$.
      In (e), we observe \{$\A\B\C$\}, and this corresponds to a stable equilibrium point, or a sink. 
    }
    \label{fig:phase}
\end{figure}

\section{Bifurcation analysis of the three-oscillator model}
\label{sec:3body-critical}
To understand the mechanism of the formation and switching of community herein,
we perform a bifurcation analysis and derive the critical coupling strengths
to make the phase diagram shown in Fig.~\ref{fig:3body-phase}.

As shown below, the $K_4$ derivation is straightforward because it is associated with the instability of the fixed point that corresponds to the full synchrony in the reduced model \eqref{eq:phase-diff}.
In contrast, the other critical coupling strengths is difficult to obtain because they are associated with the transitions between (quasi)periodic solutions. 
As suggested in Fig.~\ref{fig:3body-typical}, the transition behavior qualitatively remains even when $\nu_y$ is very small compared to $\nu_x$. Therefore, for analytical tractability, we assume that
\begin{equation}
  \eps = \frac{\nu_y}{\nu_x}
\end{equation}
is a small parameter.
Hence, we employ a perturbative approach to obtain $K_1$ and $K_2$.
Accordingly, in Fig.~\ref{fig:3body-phase}, we assume that $\nu_x=2.0$ and $\nu_y=0.05$, which is the same value used in Fig.~\ref{fig:3body-typical}.

\subsection{Summary of the transition points and the phase diagram}
\label{sec:3body-summary}
We will show that the critical coupling strengths
$K_1=\nu_x k_1, K_2=\nu_x k_2$, $K_3$, and $K_4$
are given by
\begin{subequations}
  \label{eq:3body-matome}
  \begin{align}
    %    \frac{K_1}{\nu_x} &= \frac{|\eps|}{1+s} +O(\eps^2), \label{eq:K1}\\
    %    \frac{K_2}{\nu_x} &= \sqrt{\frac{(1+r)a^2\eps +(1+s)^2+\sqrt{\left((1+r)^2a^2\eps+(1+s)^2\right)^2-\left((1+r)^2a^4-4(1+s)^2a^2\right)\eps^2}}{(1+r)^2/2\cdot a^4-2(1+s)^2\cdot a^2}}\notag\\
    %    &+O\left( \eps\left(\frac{K_2}{\nu_x}\right), \left(\frac{K_2}{\nu_x}\right)^3 \right), \label{eq:K2}\\
    k_1 = \frac{|\eps|}{1+s} +O(\eps^2), \label{eq:K1}\\
    k_2 = \frac{2(1+s)}{a^2(1+r)} + \frac{\eps}{1+s}
    +O\left( \eps^2, \eps k_2, k_2^3 \right), \label{eq:K2}\\
    K_3^\mathrm{lower} < K_3  < K_3^\mathrm{upper}, \label{eq:K3}\\
    K_4=\max\left(K_4^{(1)}, K_4^{(2)}\right), \label{eq:K4}
  \end{align}
\end{subequations}
where
\begin{subequations}
  \begin{align}
   K_3^\mathrm{lower} &= \frac{\nu_x}{(1+r)a+1},\\
   K_3^\mathrm{upper} &= \frac{\nu_x}{(1+r)a-1},\\
   K_4^{(1)}&=\frac{(1+s)\nu_x-\nu_y}{a(r+s+rs)},\\
   K_4^{(2)}&=\frac{|\nu_x-(1+r)\nu_y|}{r+s+rs}.
  \end{align}
\end{subequations}
For convenience, we regard $K_i$ as a function of $a$ (i.e., $K_i(a)$) by assuming that $\nu_x$, $\nu_y$, $r$ and $s$ are fixed.

\begin{figure}[tbp]
  \centering
  \includegraphics[width=135mm]{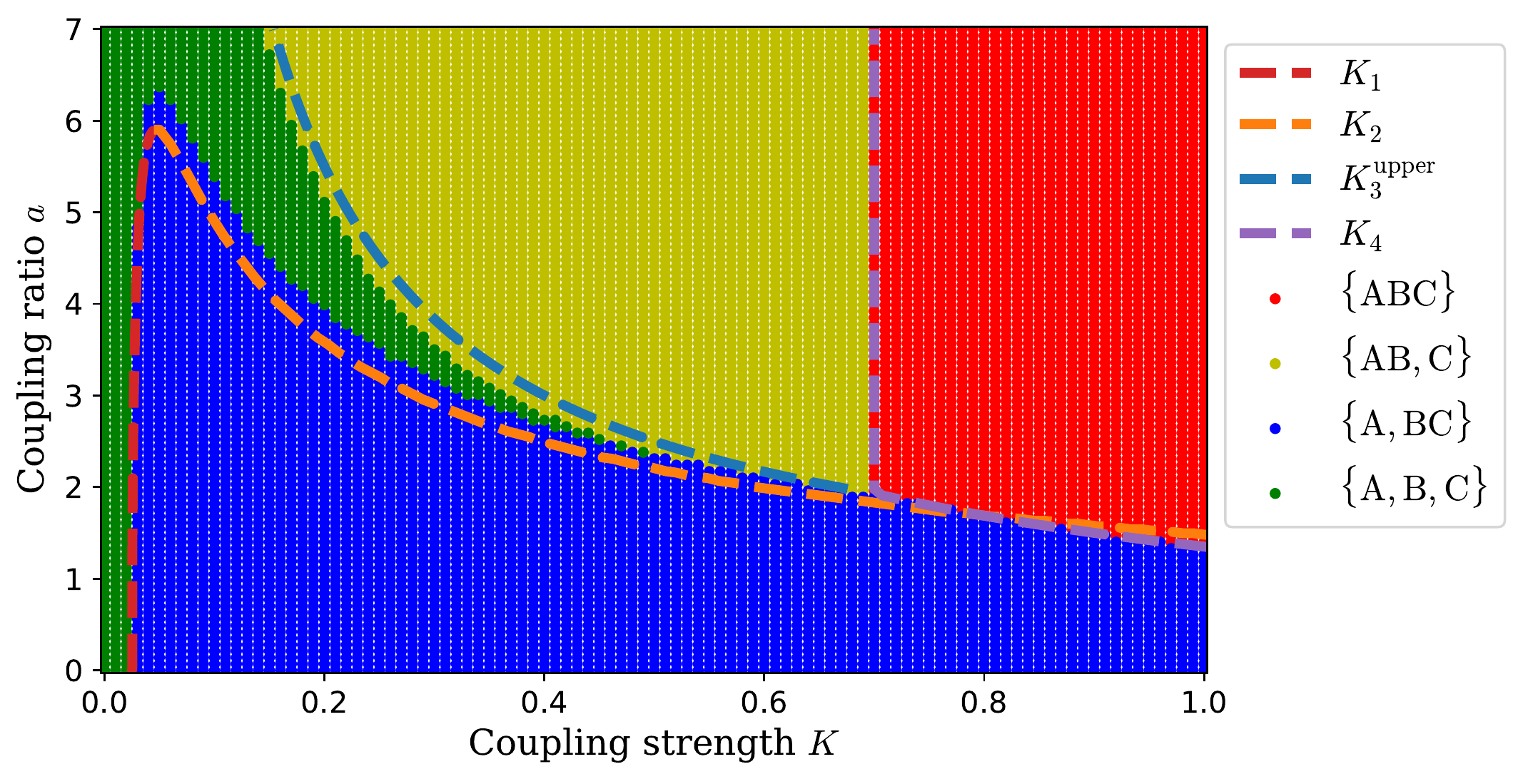}
  \centering
  \caption{Phase diagram of the synchronization patterns for $\nu_x=-2,\ \nu_y=0.05,\ r=1,\ s=1$ in the $K$-$a$ plane. 
    Each grid point shows one of the following four patterns: $\{\mathrm{ABC}\},\ \{\mathrm{AB,C}\},\ \{\mathrm{A,BC}\}$ and $\{\mathrm{A,B,C}\}$.
    The dashed curves denote the approximate critical coupling strengths.
    Whereas $K_3^\mathrm{upper}$ and $K_4$ are those given in Eqs.~\eqref{eq:3body-matome}, 
    $K_1=\nu_x k_1$ and $K_2=\nu_x k_2$ are higher-order approximations than
    Eqs.~\eqref{eq:3body-matome}, as detailed in Sec.~\ref{sec:pertu}.
  }
  \label{fig:3body-phase}
\end{figure}

Figure~\ref{fig:3body-phase} displays the phase diagram in the $K$-$a$ plane.
The phase boundaries are approximated using $K_i(a)$, where we use $K_3=K_3^\mathrm{upper}$.
In this, the $K_i$ curves are drawn only for a range of $a$ in which $K_1 \le K_2 \le K_3 \le K_4$ holds true.

From the phase diagram, we can identify the $a$ values at which the transition behaviors change. Let $(K^\dagger,a^\dagger)$ be the point where $K_3$ and $K_4$ meet and $(\hat K, \hat a)$ be the point where $K_1$ and $K_2$ meet. 
To approximately find $a^\dagger$, we solve $K_3^\mathrm{upper}=K^4$ in terms of $a$ to obtain
\begin{equation}
  a^\dagger \approx \frac{(1+s)\nu_x-\nu_y}{|\nu_x-(1+r)\nu_y|}.
\end{equation}
As detailed in Sec.~\ref{sec:pertu}, we also obtain
\begin{equation}
  \hat a \approx \frac{1+s}{\sqrt{-2(1+r)\eps}}.
\end{equation}
Therefore, the transition behaviors shown in Figs.~\ref{fig:3body-abs}(a--c) are observed approximately for $a<\hat a$, $\hat a<a<a^\dagger$, and $a>a^\dagger$, respectively.

We expect that $K_2$ also meets $K_4$ at point $(K^\dagger,a^\dagger)$ because the transition between patterns $\{\A, \B\C\}$ and $\{\A\B, \C\}$ is possible only via pattern $\{\A, \B, \C\}$, as geometrically argued in Sec.~\ref{sec:k3}. 
We also observe that for $\nu_x\geq (1+r)\nu_y$, $K_3^\mathrm{upper}$ meets $K_4$ at the point where $K_4^{(1)}$ and $K_4^{(2)}$ meet. We conjecture that $(K^\dagger,a^\dagger)$ is actually given by the point where $K_4^{(1)}$ and $K_4^{(2)}$ meet.
The simulation result shown in Fig.~\ref{fig:3body-phase} supports this conjecture.

In addition, we provide qualitative implications of these theoretical results.
The following descriptions are based on the expressions for the critical coupling strengths given in Eq.~\eqref{eq:3body-matome} and the correspondence between our three-oscillator model and the oscillator networks illustrated in Fig.~\ref{fig:model}.
With a sufficiently weak coupling, as described in $k_1$, group B is either independent or synchronized to group C, which is the group with a close natural frequency.
For a stronger coupling, as described in $k_2$ and $K_3$, group B may be desynchronized with group C and instead synchronized to group A, which is the group with a distant natural frequency.
This is more likely to occur when $a$ and/or $r$ is larger, that is, the connection between groups A and B is strong, and/or group A is small.

These implications may be further reworded in the following manner:
With a sufficiently weak coupling, the intrinsic dynamical property (i.e., the configuration of the natural frequencies in our models) has a strong effect on the community formation.
In contrast, for a strong coupling, the static property of the network becomes dominant.
Therefore, in the presence of a discrepancy between the dynamical and static properties, community switching generally occurs when the overall coupling strength is varied. 

\subsection{Perturbative approach for obtaining $K_1$ and $K_2$}
\label{sec:pertu}
We analyze transitions $\{\A,\B,\C\}$ $\to$ $\{\A,\B\C\}$ $\to$ $\{\A,\B,\C\}$ arising at $K=K_1$ and $K_2$.
As illustrated in Fig.~\ref{fig:3body-typical}, these transitions occur for small $K$ when $\nu_y$ is small.
Thus, We take a perturbative approach by assuming that  
\begin{subequations}
  \begin{align}
    \eps=\frac{\nu_y}{\nu_x} \ll 1,\\
   k=\frac{K}{\nu_x} \ll 1.
  \end{align}
\end{subequations}
In this case, there is a time-scale separation in Eq.~\eqref{eq:phase-diff}.
In other words, $\hat x = x-\nu_x t$ and $y$ are slow variables, and 
$\sin x=\sin(\hat x + \nu_x t)$ describes fast oscillations.
Thus, One can expect that the $\sin x$ terms in Eq.~\eqref{eq:phase-diff} are effectively averaged out, and both $\dot {\hat x}$ and $\dot y$ are approximately described as follows:
\begin{subequations} 
  \label{dotxy_approx}
  \begin{align}
  \dot{\hat x} &= -K\left\{ (1+r)a \sin (\hat x + \nu_x t) + \sin y \right\}\nonumber \\
  & \approx -K \sin y,\\
  \dot y &= \nu_y - K\left\{ \sin (\hat x + \nu_x t) + (1+s)a \sin y \right\} \nonumber\\
  & \approx \nu_y -2aK \sin y. \label{doty_appro}
  \end{align}
\end{subequations}
In this approximated system, the synchronization of oscillators B and C 
occurs if Eq.~\eqref{doty_appro} has a fixed point $y=y^*$( i.e., $\nu_y - (1+s)a K \sin y^*=0$).
This is the case when $K \le \frac{|\nu_y|}{(1+r)a}$ and thus we may expect $K_1 \approx \frac{|\nu_y|}{(1+r)a}$;
however, this approximation does not reproduce the second transition.
To rigorously treat such an averaging approximation up to an arbitrary order,
we apply a near-identity transformation to Eq.~\eqref{eq:phase-diff} \cite{Kuramoto2019}.
It will turn out that the above treatment actually corresponds to the lowest-order approximation, and higher order correction terms reproduce the second transition, that is, $K_2$ will be found.

First, introducing nondimensional time $\tau\coloneqq|\Omega| t$ and
\begin{subequations} 
  \label{eq:QS}
  \begin{align}
  Q(x,y)&\coloneqq -(1+r)a\sin x-\sin y, \label{eq:Q}\\
  S(x,y)&\coloneqq -a\sin x-(1+s)\sin y, \label{eq:S}
  \end{align}
\end{subequations}
we rewrite Eq.~\eqref{eq:phase-diff} as
\begin{subequations} 
  \label{eq:xy-timescaled}
  \begin{align}
  x'&=1 +k Q(x,y),\\
  y'&=\varepsilon+k S(x,y),
  \end{align}
\end{subequations}
where $x'=\frac{\mathit{d}}{\mathit{d}\tau} x(\tau)$ and $y'=\frac{\mathit{d}}{\mathit{d}\tau} y(\tau)$.
Here, we employ the following near-identity transformation $(x,y) \to (p,q)$:
\begin{subequations} 
  \label{eq:NIT}
  \begin{align}
    x &= p+\sum_{i=1}^\infty k^i g_i(p,q),\\
    y &= q+\sum_{i=1}^\infty k^i h_i(p,q),
  \label{eq:pertu-ytilde}
  \end{align}
\end{subequations}
where $g_i$ and $h_i$ are periodic functions in $p$ and $q$; i.e., $g_i(p+2 m \pi,q + 2 n \pi)=g_i(p,q)$ and $h_i(p+2 m \pi,q + 2 n \pi)=h_i(p,q)$ for $n,m \in \mathbb Z$.
Without loss of generality, we assume
\begin{subequations} 
  \begin{align}
    g_i(0,\cdot)&=0, \\
    h_i(0,\cdot)&=0. 
  \end{align}
\end{subequations}
The new variables $(p,q)$ well approximate $(x,y)$ for small $k$.

We aim to seek the functions $g_i$ and $h_i$ that transform Eq.~\eqref{eq:xy-timescaled} into 
\begin{subequations} 
  \label{eq:diffpq-expand}
  \begin{align}
    p'&=1+\sum_{i=1}^\infty k^i G_i(q)\label{eq:diffp-expand},\\    
    q'&=\varepsilon+\sum_{i=1}^\infty k^i H_i(q)\label{eq:diffq-expand},    
  \end{align}
\end{subequations}
where we assume that the right-hand equations are free from $p$ similar to Eq.~\eqref{dotxy_approx}.
Thus, the synchronization of oscillators $\B$ and $\C$ corresponds to a fixed point in Eq.~\eqref{eq:diffq-expand}.
As detained in Sec.~\ref{appendix:k1k2}, we obtain 
\begin{subequations} 
  \label{pq-dash}
  \begin{align}
    p'&=1-k\sin q +k^2\left(a\cos q -\frac{\left( 1+r \right)^2a^2}{2}\right) +O(\eps k,k^3)\label{p-dash}\\
    q'&=
    \varepsilon-k(1+s)\sin q-k^2\left(\frac{1+r}{2}a^2-(1+s)a\cos q\right)\notag\\
    &+k^3\left( \frac{(1+r)^2}{4} a^3-\frac{a^2}{4}\cos q-\frac{1+s}{2}a\sin^2 q-(1+s)a\sin q\cos q \right)+O(\eps k,k^4) \label{q-dash}
  \end{align}
\end{subequations}
The synchronization pattern $\{\A,\B\C\}$ is realized when there is a solution in $q'=0$. 
In Fig.~\ref{fig:3body-phase}, we draw the $K_1$ and $K_2$ curves by numerically identifying the range of $k$ such that a solution in $q'=0$ exists.
Although we have derived the flow of $q'$ up to $O(k^3)$ to precisely calculate the critical coupling strengths, we have the drawback of possibly failing to obtain the explicit expression for the critical coupling strengths. Therefore, for simplicity, we use Eq.~\eqref{q-dash} up to $O(k^2)$ and rewrite
\begin{equation}
  q' = \varepsilon +  (1+s) \cos(q+C)k  - \frac{(1+r)a^2}{2} k^2 + O(\eps k, k^3). 
  \label{q-dash2}
\end{equation}
where $C$ is a constant.
Note that we use $\sqrt{1 + (ak)^2}=1+O(k^2)$.
Now, solving $q'=0$ and using $|\cos (q+C)|\le 1$, we obtain the condition for the existence of \{A,BC\} as follows:
\begin{align}
  k_1 \leq k \leq k_2,
\end{align}
where $k_1$ and $k_2$ are given in Eq.~\eqref{eq:3body-matome}.

Both transitions occurring at $K_1$ and $K_2$ correspond to the saddle-node bifurcation of a pair of fixed points and of limit cycles in the systems described by $(p,q)$ and $(x,y)$, respectively.

\subsection{Existence and stability analysis of full synchrony for obtaining $K_4$}
\label{sec:k4}
Full synchrony $\{\A\B\C\}$ is realized when there is a stable fixed point in Eq.~\eqref{eq:phase-diff}. 
The fixed points
\begin{equation*}
    (x,y)=(x_1,y_1),\ (x_1,y_2),\ (x_2,y_1),\ (x_2,y_2),
\end{equation*}
exist if and only if
\begin{equation}
    K\geq K_4,
\end{equation}
where
\begin{subequations} 
  \begin{align}
    -\frac{\pi}{2} \le  x_1 & = \arcsin \frac{(1+s)\nu_x-\nu_y}{(r+s+rs)aK} \le \frac{\pi}{2},\\
    -\frac{\pi}{2} \le  y_1 & = \arcsin \frac{(1+r)\nu_y-\nu_x}{(r+s+rs)K} \le \frac{\pi}{2},\\
    x_2 &= \pi - x_1,\\
    y_2 &= \pi - y_2.
  \end{align}
\end{subequations}
These points correspond to two and four different fixed points for $K=K_4$ and $K>K_4$, respectively.
For $K > K_4$, the four fixed points correspond to a sink, a source, and two saddles as shown in Appendix.~\ref{appendix:k4}.
A full synchrony particularly arises for $K > K_4$ because a sink exists. 

\subsection{Geometric approach for obtaining the upper and lower bounds of $K_3$}
\label{sec:k3}
A synchronization of oscillators A and B occurs when $x$ is bounded.
The sufficient condition for $x$ to be bounded is that there exist such constants $x_1$ and $x_2$ that $\dot x|_{x=x_1}\ge 0$ and $\dot x|_{x=x_2}\le 0$ holds true for all $y$ in Eq.~\eqref{eq:phase-diff-x}
because then, any trajectory may not cross sections $x=x_1$ and $x=x_2$.
Using the facts that $\nu_x>0$, $\sin x_1$ takes the minimum at $x_1=-\frac{\pi}{2}$, and $\max_y \sin y = 1$, we find that $\dot x|_{x=-\frac{\pi}{2}}\ge 0$ for all $y$ when $K\ge K_3^\mathrm{upper}=\frac{\nu_x}{a(1+r)-1}$. Additionally, $\dot x|_{x=\frac{\pi}{2}}\le 0$ also holds true when $K\ge K_3^\mathrm{upper}$. Therefore, the synchronization of oscillators A and B is guaranteed for $K\ge K_3^\mathrm{upper}$.
Moreover, as proven in Sec.~\ref{appendix:k3}, the following proposition holds true:
\begin{proposition}
\label{prop:suff}
The system given by Eq.~\eqref{eq:phase-diff} with $K_3^\mathrm{upper} \leq K < K_4$
has stable and unstable periodic orbits
in regions $D=\left[-\frac{\pi}{2}, \frac{\pi}{2}\right] \times (-\pi,\pi ]\subset\mathbb{T}^2$ and $D^\mathsf{c}$, respectively.
These periodic orbits have zero and one winding numbers with respect to the $x$ and $y$ directions, respectively. 
Any other type of periodic orbit does not exist in the entire phase space $\mathbb{T}^2$. 
\end{proposition}

Note that along the periodic orbits, $x$ and $y$ are bounded and unbounded, respectively; thus, these periodic orbits correspond to pattern $\{\A\B, \C\}$ instead of $\{\A\B\C\}$.
No fixed point exists for $K<K_4$, and no other type of periodic orbits can be found; hence, the proposition ensures that for $K_3^\mathrm{upper} \leq K < K_4$, pattern $\{\A\B, \C\}$ is realized as $t\to \infty$ for any initial conditions. 

Note also that for the existence of a stable periodic orbit, instead of $D$, we can choose a narrower region $\mathcal{D}$. Let $x_\mathrm{null}(y)$ be a $x$-nullcline given by
\begin{equation}
  x_\mathrm{null}(y) = \arcsin \left(\frac{1}{a(1+r)} \left(  \frac{\nu_x}{K}-\sin y \right)\right).  
\end{equation}
We can then take $x_1$ and $x_2$ to the minimum and maximum of the $x$-nullcline $x_\mathrm{null}(y)$, respectively,
i.e., $\mathcal{D}=\left[x_1, x_2\right] \times (-\pi,\pi ]$,
where $x_1=x_\mathrm{null}(-\frac{\pi}{2})$, $x_2=x_\mathrm{null}(\frac{\pi}{2})$.

Next, we consider the condition for $x$ to be unbounded. Its sufficient condition is $\dot x>0$ for all $x$ and $y$ in Eq.~\eqref{eq:phase-diff-y}. This condition is satisfied for $K<K_3^\mathrm{lower}=\frac{\nu_x}{2a-1}$.
Therefore, the critical coupling strength $K_3$, at which oscillators A and B begin to synchronize, is given in the range $K_3^\mathrm{lower} \le K_3 < K_3^\mathrm{upper}$. 

The numerical simulation in Fig.~\ref{fig:phase}(d) indicates that the system approaches an isolated stable periodic orbit (i.e., a stable limit-cycle). 
Under the assumption that periodic orbits existent for $K_3 < K < K_3^\mathrm{upper}$ are limit cycles, the bifurcation type at $K=K_3$ is determined as follows:
First, the limit cycles have zero winding number along the $x$ direction.
On the one hand, as far as those limit cycles exist, $x$ is bounded because any trajectory may not cross the periodic orbits.
On the other hand, for $K<K_3$, $x$ is unbounded.
This implies that the limit cycles disappear at $K=K_3$. 
No fixed point exists for $K<K_4$; hence, the limit cycles may disappear only via pair annihilation; i.e., the saddle node bifurcation of limit cycles occurs at $K=K_3$.
This also implies that as $K$ decreases, pattern $\{\A,\B,\C\}$ must arise just below $K=K_3$ until $K=K_2$, at which pattern $\{\A,\B\C\}$ arises. 

Yet, whether or not these limit cycles exist for $K>K_4$ remains unclear.
If the emergence of the fixed points occurs on the limit cycles at $K=K_4$ (i.e., a saddle-node on invariant cycle (SNIC) bifurcation occurs), only pattern $\{\A\B\C\}$ arises for $K>K_4$; otherwise, there is a region of bistability for $K>K_4$ in which both patterns $\{\A\B\C\}$ and $\{\A\B,\C\}$ are stable. 
In our numerical simulations, the bistability was not observed 
for $K>K_4$, suggesting that SNIC bifurcation actually occurs at $K=K_4$. 

\section{Application of the three-oscillator model to oscillator networks}
\label{sec:applicationEg}
We herein apply our theoretical results to an example network and demonstrate the unity of our theory.

We consider the network presented in Fig.~\ref{fig:3body-application}.
This network is the same network that used in Fig.~\ref{fig:example}(a) for $l=7$ and Fig.~\ref{fig:example}(b).
\begin{figure}[tbp]
    \centering
    \includegraphics[width=80mm]{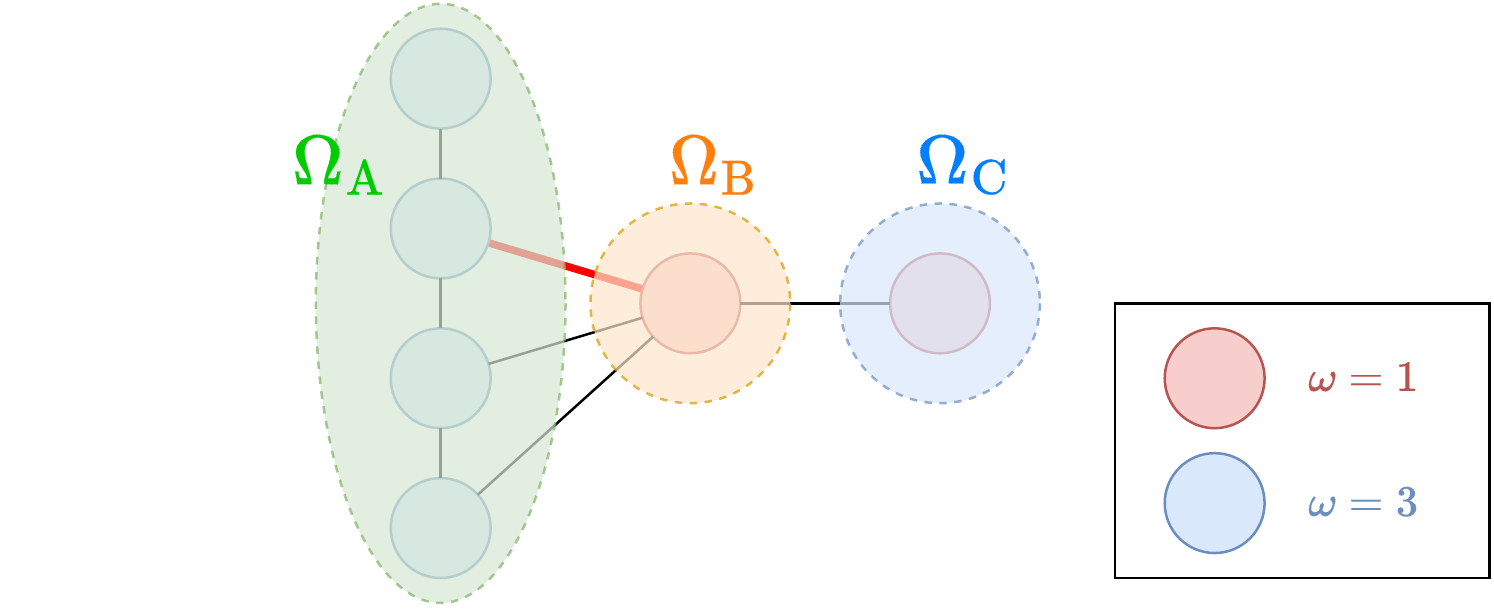}
    \centering
    \caption{Network composed of six oscillators with seven bidirectional edges.
    There are three groups A, B, and C, and their constituent oscillators have the natural frequencies $\omega_\A=3$ in group A and $\omega_\B=\omega_\C=1$ in groups B and C. 
    }
    \label{fig:3body-application}
\end{figure}

In this network, we have $\omega_\A=3, \omega_\B=\omega_\C=1, M_\A=4, M_\B=M_\C=1, m=3$, and $n=1$; thus, the parameters in Eq.~\eqref{eq:phase-diff} are determined as $\nu_x=2, \nu_y=0, a=3,\ r=\frac{1}{4}$, and $s=1$. Note that $\lambda=\frac{M_\B}{n}K=1$ in this case.

Substituting these parameter values into Eq.~\eqref{eq:3body-matome}, 
we obtain the critical coupling strengths as $K_1=0,\ K_2\approx 0.5,\ K_3^\mathrm{upper}\approx 0.7,\ K_4=\frac{4}{3}$.
Consequently, noting the scaling of the coupling strength in Eq.~\eqref{eq:phase-diff}, the network shows four synchronization patterns:
(i) $\{\mathrm{A,BC}\}$ for $0\leq K\lesssim 0.5$,
(ii) $\{\mathrm{A,B,C}\}$ for $0.5\lesssim K\lesssim 0.7$,
(iii) $\{\mathrm{AB,C}\}$ for $0.7\lesssim K\lesssim \frac{4}{3}$,
(iv) $\{\mathrm{ABC}\}$ for $K\geq \frac{4}{3}$.
Thus, our theoretical prediction is in a reasonable agreement with the numerical result presented in Fig.~\ref{fig:changingK_N6}.

Subsequently, to access the effect of network connectivity, we consider the different numbers $m$ of edges between groups A and B. 
We have $m=4$ when all edges between A and B exist, corresponding to $a=4$.
By pruning the edges from the top, we obtain $m=a=3,2,1,0$.
For example, the network shown in Fig.~\ref{fig:3body-application} in the absence of the red edge corresponds to $m=a=2$.
In Fig.~\ref{fig:3body-phase-M4N1}, we draw the critical coupling strengths predicted from the three-oscillator model as the functions of $a$ by dashed lines.
Moreover, we plot by solid lines the synchronization patterns obtained from the numerical simulations of the oscillator network.
The $K_4$ curve well approximates the phase boundary and the $K_3^\mathrm{upper}$ curve also works as an upper boundary for pattern $\{\A,\B,\C\}$.
It should be noted that $K_3^\mathrm{upper}$ and $K_4$ meet at $a=2$, predicting that patterns $\{\A\B, \C\}$ and $\{\A,\B,\C\}$ vanish around $a=2$.
Actually, these patterns, which are present for $a=4$ and $3$, are absent for $a \leq 2$; hence, community switching does not occur $a \leq 2$.
We also observe that the $K_2$ curve roughly approximates the boundary between patterns $\{\A,\B\C\}$ and $\{\A,\B,\C\}$ for $a=3$ and $4$.
However, for $a \leq 2$, the $K_2$ curve is inaccurate, as naturally expected from the fact that $K_2$ is valid only when it is sufficiently small.
This theoretical prediction is also in a reasonable agreement with the numerical result presented in Fig.~\ref{fig:cutting_N6K10}.

\begin{figure}[tbp]
  \centering
  \includegraphics[width=105mm]{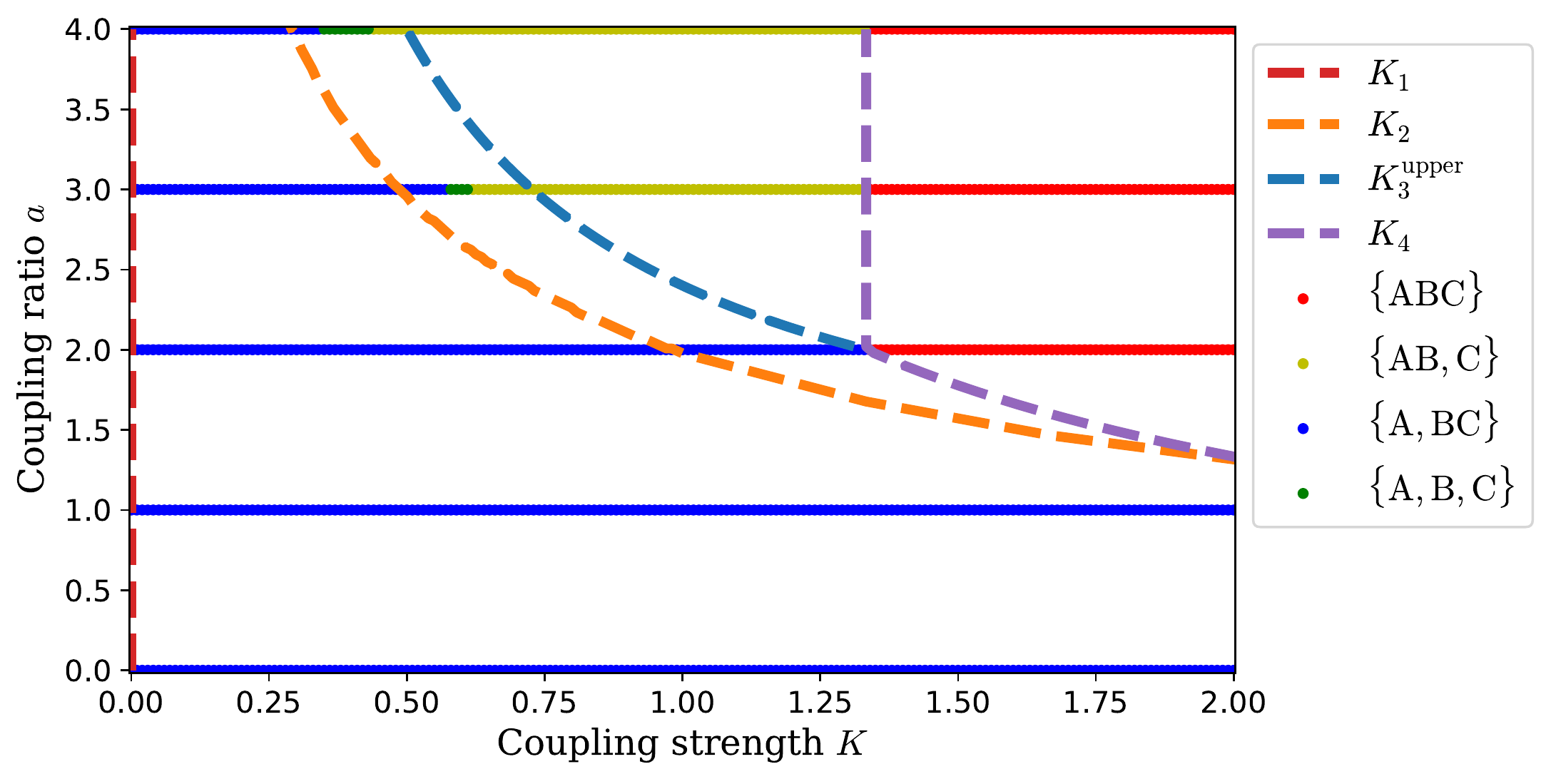}
  \caption{
    Synchronization patterns in the $K$-$a$ plane for fixed $\nu_x=2,\ \nu_y=0,\ r=\frac{1}{4}$ and $s=1$. 
    Each grid point shows one of the following four patterns: \{$\A\B\C$\}, \{$\A\B,\C$\}, \{$\A,\B\C$\} and \{$\A,\B,\C$\}.
    As well as Fig.\ref{fig:3body-phase}, the approximate critical coupling strengths are also drawn by dashed lines.
  }
  \label{fig:3body-phase-M4N1}
\end{figure}

\section{Conclusion and discussion}
\label{sec:concludediscuss}
We analytically and numerically investigated the development of dynamical community structure, where the community is referred to as a group of oscillators synchronized in frequency, in networks of phase oscillators.
We first numerically showed that in the example networks, the transitions between different synchronization patterns occur when the overall coupling strength or the network connectivity is varied. 
In particular, we found that community switching (i.e., a portion of oscillators change the group to which they synchronize) occurs for a range of parameters.
We then proposed a three-oscillator model as a mean-field approximation of a class of oscillator networks, where the connection weights of the three-oscillator models are determined by the original network connectivity. 
Using this three-oscillator model, we analyzed the transitions between different synchronization patterns and obtained approximate expressions for the critical coupling strengths in terms of natural frequencies and connection weights.
In particular, we elucidated that while oscillators with similar natural frequencies tend to synchronize for weak coupling, 
tightly connected oscillators tend to synchronize for strong coupling.
Therefore, if a discrepancy exists between the dynamical and static properties in oscillator networks, community switching generally occurs when the overall coupling strength is varied.
Using an example network, we finally confirmed that the critical coupling strengths theoretically obtained from the three-oscillator model and numerically acquired from the example network are in a rough agreement, demonstrating the utility of our three-oscillator model for predicting the development of communities in a class of complex networks.

We will now comment on the technical contribution of this study.
Partially synchronized patterns correspond to the periodic solutions of the system, the explicit expressions for which are generally difficult to obtain. 
This hinders the theoretical treatment of periodic solutions, including the existence and stability analysis.
However, our systematic perturbative approach based on an assumption on the frequencies allows us to derive a reduced model, in which the partially synchronized state appearing for a weak coupling is described as a fixed point. This renders its existence and stability analysis possible.
Note that the first-order approximation, which can be done even without an elaborate formation, can only explain the first transition where the partially synchronized state emerges.
To clarify the second transition, where the partially synchronized state disappears again, we must go beyond the first-order approximation, which requires formulation using the near-identity transformation.
As clearly seen in Eq.~\eqref{q-dash2}, the second-order term effectively shifts the frequency; thus, the frequency synchronization is violated via a saddle-node bifurcation for a larger coupling $k$. 

Our three-oscillator model is rather general in a sense that 
it describes bidirectional interactions with arbitrary weights between neighboring oscillators, which may have different natural frequencies. 
As argued in the present paper, our three-oscillator model also approximately describes a class of complex networks.
We thus expect its broad applicability, including module networks \cite{71c200abea92433b91489a97449f2bdc} and three-layered oscillatory systems, such as the system of solar-cycle oscillation \cite{Blanter2018}.

\begin{acknowledgments}
  H.K. acknowledges support from JSPS KAKENHI under grant no. 21K12056.
\end{acknowledgments}

\appendix

\section{Networks used in Sec. \ref{sec:method-3body-settting}}
In Fig.~\ref{fig:cutting_N6K10}, the network evolved in the following manner: 
We started from the complete network, where the number $l$ of edges was $l=15$.
We then removed the bidirectional edges one by one in the following order:
$(3, 5)$, $(0, 2)$, $(0, 5)$, $(0, 4)$, $(2, 5)$, $(2, 4)$, $(0, 3)$, $(1, 2)$, $(1, 3)$, $(1, 4)$, $(1, 5)$, $(0, 1)$, $(2, 3)$, $(3, 4)$, and $(4, 5)$,
where $(i,j)$ denotes the bidirectional edge between oscillators $i$ and $j$ $(0 \leq i,j \leq 5)$.
The final network is the empty network $l=0$.

The adjacency matrix used in Fig.~\ref{fig:changingK_N6} is
\begin{equation}
  \label{eq:adjmat}
  A=\begin{pmatrix}
    0&1&0&0&0&0\\
    1&0&0&1&1&1\\
    0&0&0&1&0&0\\
    0&1&1&0&1&0\\
    0&1&0&1&0&1\\
    0&1&0&0&1&0\\
  \end{pmatrix},
\end{equation}
which represents the same network considered in Fig.~\ref{fig:cutting_N6K10} with $l=7$ and the same one in Fig.~\ref{fig:3body-application} in the presence of the edge colored in red.

\section{$m$-$n$ synchronization in three-oscillator model}
\label{appendix:mn-sync}
We observed the non-monotonic dependence of the effective frequencies on the coupling strength in Sec.\ref{sec:3body-synchro-state}.
We hypothesized that this was caused by $m$-$n$ synchronization between $x$ and $y$; i.e.,
$x$ increases by $2m\pi$ while $y$ increases by $2 n \pi$ ($m,n\in\mathbb{Z}$).
To confirm this, as shown in Fig.~\ref{fig:m-n_sync},
we plot the ratio $\zeta(K)$ between the effective frequencies of $x$ and $y$ as
\begin{equation}
    \zeta 
    = \frac{\Omega_1^\mathrm{eff} - \Omega_2^\mathrm{eff}}
    {\Omega_3^\mathrm{eff} - \Omega_2^\mathrm{eff}}.
\end{equation}
Here, we can easily find plateaus in $\zeta(K)=m$, where $m=0,1,\ldots, 6$, indicating that the oscillations in $x$ and $y$ causes $m$-$1$ synchronization.
Moreover, we can also observe other plateaus, where $\zeta(K)=\frac{m}{n}$ ($m,n\in\mathbb{Z}$).
Hence, the curve might be considered as the devil's stair case, which is observed when $m$-$n$ synchronization occurs\cite{pikovsky_rosenblum_kurths_2001}

This behavior is also observed in a similar oscillator network \cite{SUCHORSKY20092434}.

\begin{figure}[tbp]
  \centering
  \includegraphics[width=80mm]{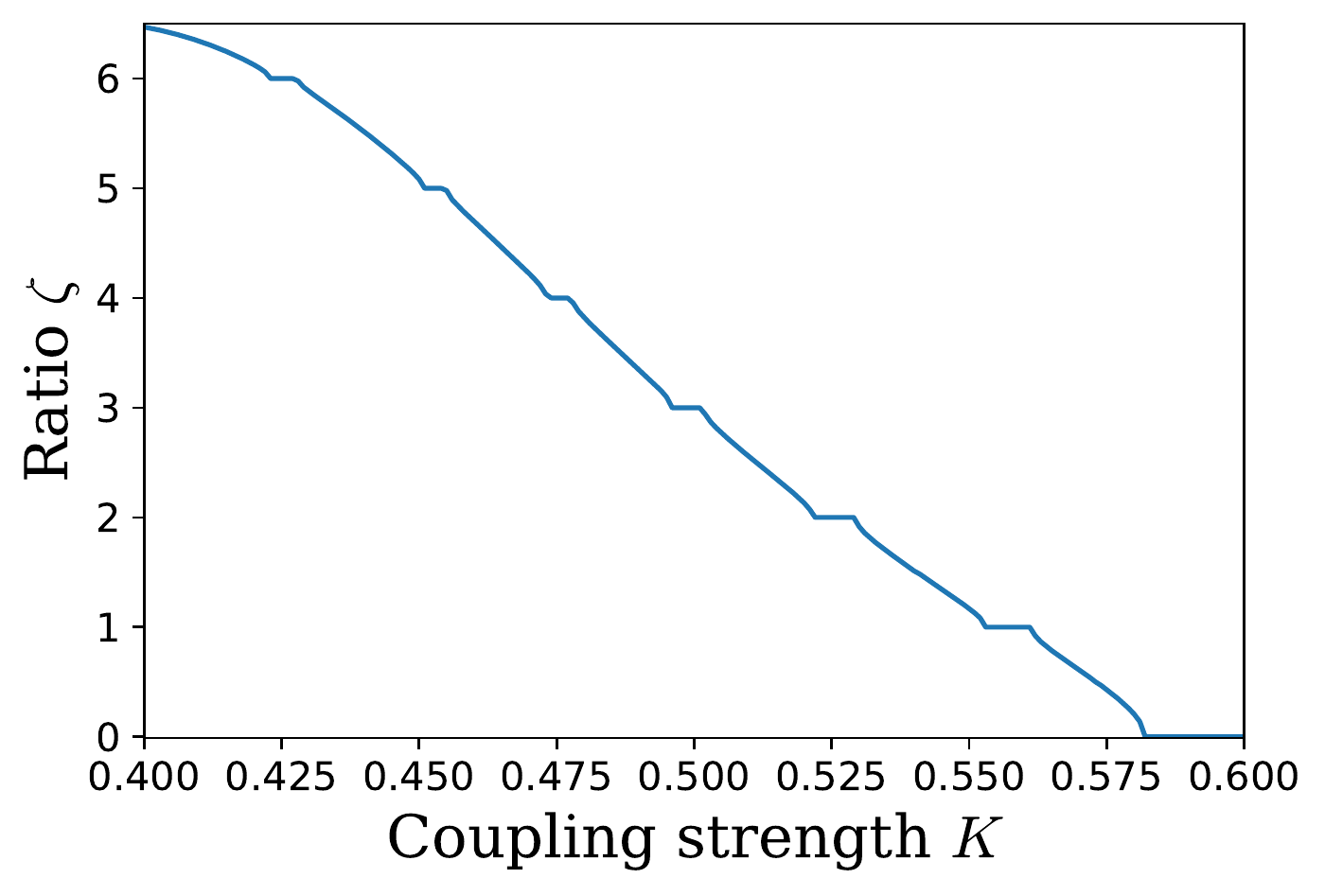}
  \caption{Ratio $\zeta$ of the effective frequencies of $x$ and $y$ numerically observed as a function of $K$ for $a=2$, $\nu_x=2$, $\nu_y=-0.5$ and $r=s=1$, which is the same parameter values as in Fig.~\ref{fig:3body-abs}(c).
  The curve monotonically decreases and has some plateaus, where the ratio is rational.
  }
  \label{fig:m-n_sync}
\end{figure}

\section{Derivation of $K_1$ and $K_2$}
\label{appendix:k1k2}
In this section, we derive the functions $H_i(q),\ G_j(q)$ ($i=1,\ 2$, $j=1,\ 2,\ 3$) in Eqs.~\eqref{eq:diffp-expand} and \eqref{eq:diffq-expand}.

Substituting Eq.~\eqref{eq:NIT} into Eq.~\eqref{eq:QS}, we expand $Q(x,y),\ S(x,y)$ as follows with respect to $p$ and $q$:
\begin{subequations} 
  \begin{align}
    Q(x,y)&=Q_1(p,q)+kQ_2(p,q)+O(k^2),\\
    S(x,y)&=S_1(p,q)+kS_2(p,q)+k^2S_3(p,q)+O(k^3),
  \end{align}
\end{subequations}
where
\begin{align}
  Q_1&=-(1+r)a\sin p-\sin q,\\
  Q_2&=-(1+r)a\cos p\cdot g_1-\cos q\cdot h_1,\\
  S_1&=-a\sin p-(1+s)\sin q,\\
  S_2&=-a\cos p\cdot g_1-(1+s)\cos q\cdot h_1,\\
  S_3&= \frac{a}{2}\sin p\cdot g_2+\frac{1+s}{2}\sin q\cdot h_2.
\end{align}

Differentiating Eq.~\eqref{eq:NIT} with respect to $\tau$,
we obtain
\begin{align}
  x'&=p'+k(p'\partial_p g_1+q'\partial_q g_1)+k^2(p'\partial_p g_2+q'\partial_q g_2)+O(k^3)\notag\\
  &=(1+k\partial_p g_1+k^2\partial_p g_2)p'+(k\pt_q g_1+k^2\pt_q g_2)q'+O(k^3),\label{eq:diff-x-pq}\\
  y'&=(1+k\partial_q h_1+k^2\partial_q h_2+k^3\partial_q h_3)q'+(k\pt_p h_1+k^2\pt_p h_2+k^3\pt_p h_3)p'+O(k^4),\label{eq:diff-y-pq}
\end{align}
where $\ast'\coloneqq\dd{\ast}/\dd{\tau}$, $\partial_p \ast \coloneqq\partial{\ast}/\partial p$, and $\partial_q \ast \coloneqq\partial{\ast}/\partial q$. 
Rearranging Eq.~\eqref{eq:diff-x-pq}, we obtain
\begin{align}
  p'&=(1+k\pt_p g_1+k^2\pt_p g_2)^{-1}\left( x'-(k\pt_q g_1+k^2\pt_q g_2)q' \right)+O(k^3)\notag\\
  &=(1-k\pt_p g_1-k^2\pt_p g_2+ k^2 (\pt_p g_1)^2)\left(\operatorname{sgn}(\Omega)+kQ_1+k^2 Q_2-(k\pt_q g_1+k^2\pt_q g_2)q'\right)+O(k^3).\label{eq:diff-p-order2}
\end{align}
Substituting Eq.~\eqref{eq:diff-p-order2} into Eq.~\eqref{eq:diff-y-pq}, we obtain
\begin{align}
  y'&=(1+k\partial_q h_1+k^2\partial_q h_2+k^3\partial_q h_3)q'+(k\pt_p h_1+k^2\pt_p h_2+k^3\pt_p h_3)(1-k\pt_p g_1-k^2\pt_p g_2+ k^2 (\pt_p g_1)^2)\notag\\
  &\quad\left(\operatorname{sgn}(\Omega)+kQ_1+k^2 Q_2-(k\pt_q g_1+k^2\pt_q g_2)q'\right)+O(k^4)\\
  &=(1+k\partial_q h_1+k^2\partial_q h_2+k^3\partial_q h_3)q'\notag\\
  &+\left(k\pt_p h_1+k^2(-\pt_p h_1\cdot \pt_p g_1+\pt_p h_2)+k^3(\pt_p h_3-\pt_p h_1\cdot\pt_p g_2+\pt_p h_1\cdot (\pt_p g_1)^2-\pt_p h_2\cdot \pt_p g_1\right)\notag\\
  &\quad\left(\operatorname{sgn}(\Omega)+kQ_1-k\pt_q g_1q'+k^2Q_2-k^2\pt_q g_2q'\right)+O(k^4)\\
  &=\left(1+k\partial_q h_1+k^2(\partial_q h_2-\pt_p h_1\cdot \pt_q g_1)+k^3\left(\pt_q h_3+\pt_qg_1\cdot(-\pt_ph_2+\pt_ph_1\cdot\pt_pg_1)-\pt_qg_2\cdot\pt_ph_1\right)\right)q'\notag\\
  &+\left(k\pt_p h_1+k^2(-\pt_p h_1\cdot \pt_p g_1+\pt_p h_2)+k^3(\pt_q h_3-\pt_p h_1\cdot\pt_p g_2-\pt_p h_2\cdot \pt_p g_1+\pt_p h_1\cdot(\pt_p g_1)^2\right)\notag\\
  &\quad\left(\operatorname{sgn}(\Omega)+kQ_1+k^2Q_2\right)+O(k^4).
\end{align}
Thus, we obtain
\begin{align}
  q'&=\left(1+k\partial_q h_1+k^2(\partial_q h_2-\pt_p h_1\cdot \pt_q g_1)+k^3\left(\pt_q h_3+\pt_qg_1\cdot(-\pt_ph_2+\pt_ph_1\cdot\pt_pg_1)-\pt_qg_2\cdot\pt_ph_1\right)\right)^{-1}\notag\\
  &\quad\ \left(y'-\left(k\pt_p h_1+k^2(-\pt_p h_1\cdot \pt_p g_1+\pt_p h_2)+k^3(\pt_p h_3-\pt_p h_1\cdot\pt_p g_2-\pt_p h_2\cdot \pt_p g_1+\pt_p h_1\cdot(\pt_p g_1)^2)\right)\right.\notag\\
  &\quad\left. \left(1+kQ_1+k^2Q_2\right)\right)+O(k^4)\\
  &=\left(1-k\partial_q h_1-k^2(\partial_q h_2-\pt_p h_1\cdot \pt_q g_1-(\pt_q h_1)^2)\right.\notag\\
  &\quad\left(-k^3\left(\pt_q h_3+\pt_qg_1\cdot(-\pt_ph_2+\pt_ph_1\cdot\pt_pg_1)-\pt_qg_2\cdot\pt_ph_1-(\pt_q h_1)^3+2\pt_qh_1\cdot(\pt_qh_2-\pt_ph_1\cdot\pt_qg_1)\right)\right)\notag\\
  &\quad\ \left(\varepsilon +kS_1+k^2S_2+k^3S_3-\left(1+kQ_1+k^2Q_2\right)\right.\notag\\
  &\left.\left(k\pt_p h_1+k^2(-\pt_p h_1\cdot \pt_p g_1+\pt_p h_2)+k^3(\pt_p h_3-\pt_p h_1\cdot\pt_p g_2-\pt_p h_2\cdot \pt_p g_1+\pt_p h_1\cdot(\pt_p g_1)^2)\right)\right)+O(k^4)\\
  &=\left(1-k\partial_q h_1-k^2(\partial_q h_2-\pt_p h_1\cdot \pt_q g_1-(\pt_q h_1)^2)\right.\notag\\
  &\quad\left(-k^3\left(\pt_q h_3+\pt_qg_1\cdot(-\pt_ph_2+\pt_ph_1\cdot\pt_pg_1)-\pt_qg_2\cdot\pt_ph_1-(\pt_q h_1)^3+2\pt_qh_1\cdot(\pt_qh_2-\pt_ph_1\cdot\pt_qg_1)\right)\right)\notag\\
  &\quad\ \left(\varepsilon +k(S_1-\pt_ph_1)+k^2\left\{S_2-Q_1\cdot\pt_ph_1+\pt_ph_1\cdot\pt_pg_1-\pt_ph_2\right\}\right.\notag\\
  &\left.+k^3\left\{S_3-Q_2\cdot \pt_ph_1-Q_1(-\pt_ph_1\cdot \pt_pg_1+\pt_ph_2)-\left(-\pt_ph_1\cdot\pt_pg_2-\pt_ph_2\cdot\pt_pg_1+\pt_p h_1\cdot(\pt_p g_1)^2+\pt_ph_3\right)\right\}\right)\\
  &=\varepsilon+k\left(-\varepsilon\pt_q h_1+S_1-\pt_p h_1\right)+k^2\left\{-\varepsilon \left(\pt_q h_2-\pt_p h_1\cdot \pt_q g_1-(\pt_q h_1)^2\right)-\pt_q h_1\left(S_1-\operatorname{sgn}(\Omega)\cdot\pt_p h_1\right)\right.\notag\\
  &\left.+S_2-\pt_p h_1\cdot Q_1+(\pt_p h_1\cdot \pt_p g_1-\pt_p h_2)\right\}\notag\\
  &+k^3\left[S_3-Q_2\cdot\pt_ph_1-Q_1\cdot\left(-\pt_ph_1\cdot\pt_pg_1+\pt_ph_2\right)\right.\notag\\
  &-(-\pt_ph_1\cdot\pt_pg_2-\pt_ph_2\cdot\pt_pg_1+\pt_p h_1\cdot (\pt_pg_1)^2+\pt_ph_3)\notag\\
  &-\pt_qh_1\cdot\left\{S_2-Q_1\cdot\pt_ph_1-(-\pt_ph_1\cdot\pt_pg_1+\pt_ph_2)\right\}\notag\\
  &-(\pt_qh_2-\pt_ph_1\cdot\pt_qg_1-(\pt_q h_1)^2)(S_1-\pt_ph_1)\notag\\
  &\left.-\varepsilon\left\{\pt_qh_3+\pt_qg_1\cdot(-\pt_ph_2+\pt_ph_1\cdot\pt_pg_1)-\pt_qg_2\cdot\pt_ph_1+2\pt_qh_1\cdot(\pt_qh_2-\pt_ph_1\cdot\pt_qg_1)-(\pt_q h_1)^3\right\}\right]+O(k^4).
  \label{eq:diffq-detail}
\end{align}
Substituting Eq.~\eqref{eq:diffq-detail} into Eq.~\eqref{eq:diff-p-order2}, we obtain
\begin{align}
  p'&=(1+k\pt_p g_1+k^2\pt_p g_2)^{-1}\left\{ 1+kQ_1+k^2Q_2-(k\pt_q g_1+k^2\pt_q g_2)q' \right\}+O(k^3)\notag\\
  &=(1-k\pt_p g_1-k^2\pt_p g_2+k^2(\pt_pg_1)^2)\notag\\
  &\quad\left[ 1+kQ_1+k^2Q_2-(k\pt_q g_1+k^2\pt_q g_2)\left\{\varepsilon+k(S_1-\pt_ph_1-\varepsilon\pt_qh_1)\right\} \right]+O(k^3)\notag\\
  &=(1-k\pt_p g_1-k^2\pt_p g_2+k^2(\pt_pg_1)^2)\notag\\
  &\quad\left[ 1+k(Q_1-\varepsilon\pt_qg_1)+k^2\left\{Q_2-\varepsilon\pt_qg_2-\pt_qg_1(S_1-\pt_ph_1-\varepsilon\pt_qh_1)\right\}\right]+O(k^3)\notag\\
  &=1+k(-\pt_p g_1+Q_1-\varepsilon \pt_q g_1)\notag\\
  &+k^2\left\{Q_2-\varepsilon\pt_qg_2-\pt_qg_1(S_1-\pt_ph_1-\varepsilon\pt_qh_1)-\pt_pg_1(Q_1-\varepsilon\pt_qg_1)\right.\notag\\
  &\left.-\left(\pt_pg_2-(\pt_pg_1)^2\right)\right\}+O(k^3).
  \label{eq:diff-p-order1}
\end{align}

Comparing Eqs.~\eqref{eq:diffq-detail} and \eqref{eq:diffq-expand} for the first order of $k$, we obtain
\begin{equation}
  H_1(q)=S_1(p,q)-\pt_p h_1(p,q)-\varepsilon \pt_q h_1(p,q).
\end{equation}
Integrating both sides from $0$ to $\tau$ $(0\leq \tau \le 2\pi)$ with respect to $p$, we obtain
\begin{align}
  \tau H_1(q)&=\int_0^{\tau}\left(-a\sin p-(1+s)\sin q\right)\dd{p}-h_1(\tau,q)+O(\varepsilon)\notag \\
  &=a(\cos \tau-1)-(1+s)\tau\sin q-h_1(\tau,q)+O(\varepsilon).
  \label{eq-h1-H1}
\end{align}
Substituting $2\pi$ into $\tau$, we obtain
\begin{align}
  H_1(q)&=-(1+s)\sin q+O(\varepsilon).
  \label{eq:H1-approx}
\end{align}
Subsequently, substituting Eq.~\eqref{eq:H1-approx} into Eq.~\eqref{eq-h1-H1} and replacing $\tau$ with $p$, we obtain
\begin{equation}
  h_1(p,q)=a(\cos p-1)+O(\varepsilon).
  \label{eq:h1-approx}
\end{equation}
Furthermore, comparing Eqs.~\eqref{eq:diff-p-order1} and \eqref{eq:diffp-expand} for the first order of $k$, we obtain
\begin{equation}
  G_1(q)=Q_1(p,q)-\pt_p g_1(p,q)-\varepsilon \pt_q g_1(p,q).
\end{equation}
Integrating both sides from $0$ to $\tau$ with respect to $p$, we obtain
\begin{align}
  \tau G_1(q)&=\int_0^{\tau}\left(-(1+r)a\sin p-\sin q\right)\dd{p}-g_1(\tau,q)+O(\varepsilon)\notag \\
  &=(1+r)a(\cos \tau-1)-\tau\sin q-g_1(\tau,q)+O(\varepsilon). 
  \label{eq-g1-G1}
\end{align}
Substituting $2\pi$ into $\tau$, we obtain
\begin{align}
  G_1(q)&=-\sin q+O(\varepsilon).
  \label{eq:G1-approx}
\end{align}
Subsequently, substituting Eq.~\eqref{eq-g1-G1} into Eq.~\eqref{eq:G1-approx}, we obtain
\begin{equation}
  g_1(p,q)=(1+r)a(\cos p-1)+O(\varepsilon).
  \label{eq:g1-approx}
\end{equation}

Next, we compare Eqs.~\eqref{eq:diffq-detail} and \eqref{eq:diffq-expand} for the second order of $k$ and utilize Eqs.\eqref{eq:h1-approx} and \eqref{eq:g1-approx} to obtain
\begin{align}
  \begin{split}
      H_2(q)=&-\varepsilon \left(\pt_q h_2-\pt_p h_1\cdot \pt_q g_1\right)-\pt_q h_1\left(S_1-\operatorname{sgn}(\Omega)\pt_p h_1\right)+S_2-\pt_p h_1\cdot Q_1+\pt_p h_1\cdot \pt_p g_1-\pt_p h_2\\
      =&-a\cos p\cdot (1+r)a(\cos p-1)-(1+s)\cos q\cdot a(\cos p-1)\\
      &-(-a\sin p)\cdot (-(1+r)a\sin p-\sin q)-a\sin p(-(1+r)a\sin p)-\pt_p h_2+O(\varepsilon)\\
      =&-(1+r)a^2\cos^2 p+a\left((1+r)a-(1+s)\cos q\right)\cdot \cos p+(1+s)a\cos q-a\sin q\sin p\notag\\
      &-\pt_p h_2+O(\varepsilon).
  \end{split}
\end{align}
Integrating both sides from $0$ to $\tau$ with respect to $p$, we obtain
\begin{align}
  \label{eq:integral-H2}
  \tau H_2(q)=&-\frac{1+r}{2}a^2\left(\tau+\frac{1}{2}\sin 2\tau\right)+a\left((1+r)a\right.\notag\\
  &\left.-(1+s)\cos q\right)\sin \tau+(1+s)a\tau\cos q+a\sin q(\cos\tau-1)-h_2+O(\varepsilon).
\end{align}
Substituting $2\pi$ into $\tau$, we obtain 
\begin{equation}
  \label{eq:H2-approx}
  H_2(q)=-a\left(\frac{1+r}{2}a-(1+s)\cos q\right)+O(\varepsilon).
\end{equation}
Substituting Eq.~\eqref{eq:H2-approx} into Eq.~\eqref{eq:integral-H2}, we acquire
\begin{equation}
  \label{eq:h2-approx}
  h_2(p,q)=-\frac{1+r}{4}a^2\sin 2p +a\left((1+r)a-(1+s)\cos q\right)\sin p+a\sin q(\cos p-1)+O(\varepsilon)
\end{equation}

Comparing Eqs.~\eqref{eq:diff-p-order1} Eq.~\eqref{eq:diffp-expand} for the second order of $k$, we obtain
\begin{align}
  G_2(q)=&Q_2-\varepsilon\pt_q g_2-\pt_q g_1(S_1-\pt_ph_1-\varepsilon\pt_qh_1)-\pt_pg_1(Q_1-\varepsilon\pt_qg_1)-\left(\pt_pg_2-(\pt_pg_1)^2\right)\notag\\
  =&-(1+r)a\cos p\cdot (1+r)a(\cos p-1)-\cos q\cdot a(\cos p-1)\notag\\
  &+(1+r)a\sin p(-(1+r)a\sin p-\sin q)-\pt_p g_2+4a^2\sin^2 p+O(\varepsilon)\notag\\
  =&-(1+r)^2a^2+\left((1+r)^2a^2-a\cos q\right)\cos p+a\cos q-(1+r)a\sin q\sin p\notag\\
  &+(1+r)^2a^2\sin^2 p+(1+r)^2a^2\left(\frac{\tau}{2}-\frac{1}{4}\sin 2\tau\right)-\pt_pg_2+O(\varepsilon).
\end{align}
Integrating both sides from $0$ to $\tau$ with respect to $p$, we obtain
\begin{align}
  \tau G_2(q)=&-(1+r)^2a^2\tau+\left((1+r)^2a^2-a\cos q\right)\sin\tau+a\tau\cos q+(1+r)a\sin q(\cos\tau-1)\notag\\
  &-g_2+O(\varepsilon\tau). 
  \label{eq-g2-G2}
\end{align}
Substituting $2\pi$ into $\tau$, we obtain
\begin{align}
  G_2(q)&=a\cos q-(1+r)^2a^2+\frac{(1+r)^2a^2}{2}+O(\varepsilon).
  \label{eq:G2-approx}
\end{align}
Subsequently, substituting Eq.~\eqref{eq-g2-G2} into Eq.~\eqref{eq:G2-approx}, we acquire
\begin{equation}
  g_2(p,q)=\left((1+r)^2a^2-a\cos q\right)\sin p+(1+s)a\sin q(\cos p-1)-\frac{(1+r)^2}{4}a^2\sin 2p+O(\varepsilon).
  \label{eq:g2-approx}
\end{equation}

We then obtain the following equation comparing Eqs.~\eqref{eq:diffq-detail} and \eqref{eq:diffq-expand} for the third order of $k$:
\begin{align}
  H_3(q)=&S_3-Q_2\cdot\pt_ph_1-Q_1\cdot\left(-\pt_ph_1\cdot\pt_pg_1+\pt_ph_2\right)\notag\\
  &-(-\pt_ph_1\cdot\pt_pg_2-\pt_ph_2\cdot\pt_pg_1+\pt_ph_1\cdot (\pt_pg_1)^2+\pt_ph_3)\notag\\
  &-\pt_qh_1\cdot\left\{S_2-Q_1\cdot\pt_ph_1-(-\pt_ph_1\cdot\pt_pg_1+\pt_ph_2)\right\}\notag\\
  &-(\pt_qh_2-\pt_ph_1\cdot\pt_qg_1-(\pt_qh_1)^2)(S_1-\pt_ph_1)\notag\\
  &-\varepsilon\left\{\pt_qh_3+\pt_qg_1\cdot(-\pt_ph_2+\pt_ph_1\cdot\pt_pg_1)-\pt_qg_2\cdot\pt_ph_1+2\pt_qh_1\cdot(\pt_qh_2-\pt_ph_1\cdot\pt_qg_1)+(\pt_qh_1)^3\right\}\\
  &=\frac{a}{2}\sin p\cdot \left\{\left((1+r)^2a^2-a\cos q\right)\sin p+(1+r)a\sin q(\cos p-1)\right\}\notag\\
  &+\frac{1+s}{2}\sin q\cdot\left\{-\frac{(1+r)}{4}a^2\sin 2p +a\left((1+r)a-(1+s)\cos q\right)\sin p+a\sin q(\cos p-1)\right\}\notag\\
  &+a\sin p\left\{-(1+r)a\cos p\cdot (1+r)a(\cos p-1)-\cos q\cdot a(\cos p-1)\right\}\notag\\
  &+\left((1+r)a\sin p+\sin q\right)\left\{a\sin p\cdot(-(1+r)a\sin p)-\frac{1+r}{2}a^2\cos 2p\right.\notag\\
  &\left.+a\left((1+r)a-(1+s)\cos q\right)\cos p-a\sin q\sin p\right\}\notag\\
  &-\left[a\sin p\left\{((1+r)^2a^2-a\cos q)\cos p-(1+r)a\sin q\sin p-\frac{(1+r)^2}{2}a^2\cos 2p-(1+r)^2a^2\sin^2p\right\}\right.\notag\\
  &\left.+(1+r)a\sin p\left\{-\frac{1+r}{2}a^2\cos 2p +a\left((1+r)a-(1+s)\cos q\right)\cos p-a\sin q\sin p\right\}+\pt_p h_3\right]\notag\\
  &-\left\{(1+s)a\sin q\sin p+a\cos q(\cos p-1)\right\}(-(1+s)\sin q)+O(\varepsilon)\\
  &=\frac{a^2}{2}\left((1+r)^2a-\cos q\right)\sin^2 p+\frac{1+r}{2}a^2\sin q\sin p\cos p-\frac{1+r}{2}a^2\sin q\sin p\notag\\
  &-\frac{(1+r)(1+s)}{8}a^2\sin q\sin 2p+a\sin q\left((1+r)a-(1+s)\cos q\right)\sin p\notag\\
  &+\frac{1+s}{2}a\sin^2 q\cos p-\frac{1+s}{2}a\sin^2 q-(1+r)^2a^3\sin^3 p-\frac{(1+r)^2}{2}a^3\cos 2p\sin p\notag\\
  &+\left\{-(1+r)^2a^3\cos^2 p\sin p+(1+r)^2a^3\sin p\cos p-a^2\cos q\sin p\cos p+a^2\cos q\sin p\right\}\notag\\
  &+(1+r)a^2\left((1+r)a-(1+s)\cos q\right)\sin p\cos p-(1+r)a^2\sin q\sin^2 p\notag\\
  &-(1+r)a^2\sin q\sin^2 p-\frac{1+r}{2}a^2\sin q\cos 2p+(1+r)a^2\sin q\sin^2 p-\pt_ph_3+O(\eps)\notag\\
  &+a\left((1+r)a-(1+s)\cos q\right)\sin q\cos p-a\sin^2 q\sin p\notag\\
  &-a^2\left((1+r)^2a-\cos q\right)\sin p\cos p+(1+r)a^2\sin q\sin^2 p\notag\\
  &+\frac{(1+r)^2}{2}a^3\cos 2p\sin p-(1+r)a^2\left((1+r)a-(1+s)\cos q\right)\sin p\cos p\notag\\
  &+(1+s)^2a\sin^2 q\sin p+(1+s)a\sin q\cos q\cos p-(1+s)a\sin q\cos q\notag\\
  &-\frac{(1+r)^2}{2}a^3\sin p\cos 2p-(1+r)^2a^3\sin^3 p.
\end{align}
Integrating both sides from $0$ to $\tau$ of $p$ and substituting $2\pi$ into $\tau$, we obtain

\begin{equation}
  H_3(q)=\frac{a^2}{4}((1+r)^2a-\cos q)-\frac{1+s}{2}a\sin^2 q-(1+s)a\sin q\cos q+O(\varepsilon).
  \label{eq:H3-approx}
\end{equation}

\section{Proof of proposition \ref{prop:suff}}
\label{appendix:k3}
This section states the proof of Proposition~\ref{prop:suff}.

First, we will show the following lemma to show Proposition~\ref{prop:suff}.
\begin{lemma}
    \label{lemma:annulus}
    On the condition of Proposition~\ref{prop:suff},
    there exist stable periodic orbits with zero and one winding number with respect to $x$ and $y$, respectively, and is homeomorphic to $\mathbb{S}^1$ in the cylindrical region $D=[x_1,x_2]\times (-\pi,\pi]\subset \mathbb{T}^2$. 
\end{lemma}
\begin{proof}
    First, the trajectory out of any point in $D$ stays in $D$.
    If there were to exist a trajectory that goes from inside $D$ to outside $D$, the trajectory would cross the $D$ boundary ($x=x_1,x_2$).
    However, $\dot{x}|_{x=x_1}>0$ and $\dot{x}|_{x=x_2}<0$ (i.e., the flow of $x$ on the boundary head inside $D$), which contradicts the existence of such trajectories.
    
    Second, no equilibrium point exists because of the assumption of Proposition~\ref{prop:suff}, and the cylindrical region $D$ is homeomorphic to an annulus but not to torus $\mathbb{T}^2$.
    Therefore, the generalization of the Poincar\'{e}-Bendixson theorem\cite{Schwartz1963ErrataAG} guarantees the existence of stable periodic orbits homeomorphic to $\mathbb{S}^1$ in $D$.
    
    Finally, we consider the winding numbers of the periodic orbits.
    In the first place, the winding number with respect to $x$ is zero because $D$ is a trapping region, where any trajectory does not go out.
    Next, if the winding number with respect to $y$ were to be zero, then the orbits would be closed in $\mathbb{R}^2$ and thus enclose fixed points by index theory\cite{Strogatz2018}.
    This contradicts the assumption that no fixed point exists.
    By contrast, if the winding number with respect to $y$ were to be more than one, then the orbits would cross themselves; this, however, contradicts the existence and uniqueness theorem\cite{Strogatz2018}.
    Therefore, the winding number with respect to $y$ is one.
\end{proof}

Next, we show the proof of Proposition~\ref{prop:suff}.
\begin{proof}
    Considering the time reversal of Lemma~\ref{lemma:annulus} and Eq.~\eqref{eq:phase-diff}, there exist stable periodic orbits and unstable ones in $D$ and $D^\mathsf{c}$ respectively.
    The trajectory out of any point then converges to the stable periodic orbits or the unstable ones by the generalization of the Poincar\'{e}-Bendixson theorem\cite{Schwartz1963ErrataAG}.
\end{proof}

\section{Stability of fixed points}
\label{appendix:k4}
In this section, we analyze the linear stability of the fixed points in Sec.\ref{sec:k4}.

We can argue the linear stabilities of the fixed points with the Jacobian matrix of Eq.~\eqref{eq:phase-diff} $J$ as follows:
\begin{align}
    J=-K\begin{pmatrix}
        (1+r)a\cos x&\cos y\\
        a\cos x&(1+s)\cos y
    \end{pmatrix}.
\end{align}
The trace and determinant of $J$ satisfy the following equations and inequality:
\begin{align}
    \tr J&=-K((1+r)a\cos x+(1+s)\cos y),\\
    \det J&=(r+s+rs)aK^2\cos x\cos y,\\
    (\tr J)^2-4\det J&=K^2\left((1+r)a\cos x-(1+s)\cos y\right)^2> 0,
\end{align}
Here, $\cos x_1>0,\ \cos x_2<0,\ \cos y_1>0,\ \cos y_2<0$ and $a>0,\ K>0$; thus,
the trace and determinant of $J$ on each fixed point satisfy
\begin{alignat*}{2}
    \left.\det J \right|_{(x,y)=(x_1,y_1)}&>0, &\quad  \left.\tr J \right|_{(x,y)=(x_1,y_1)}&<0,\\
    \left.\det J \right|_{(x,y)=(x_1,y_2)}&<0, & &\\
    \left.\det J \right|_{(x,y)=(x_2,y_1)}&<0, & &\\
    \left.\det J \right|_{(x,y)=(x_2,y_2)}&>0, &\quad  \left.\tr J \right|_{(x,y)=(x_2,y_2)}&>0.
\end{alignat*}
In conclusion, $(x_1,y_1)$ is a sink; $(x_1,y_2),\ (x_2,y_1)$ are saddles; and $(x_2,y_2)$ is a source in terms of linear stability.
Indeed, Fig.~\ref{fig:phase-k80} shows that the fixed points are composed of a sink, two saddles, and a source.

% \nocite{*}

\bibliography{article_20220831.bbl}% Produces the bibliography via BibTeX.
% \bibliography{reference.bib}% Produces the bibliography via BibTeX.

\end{document}